\documentclass[11pt,a4paper]{article}

\usepackage[left=2.5cm, right=2.5cm, top=1in, bottom=1in]{geometry}

\usepackage{amssymb}
\usepackage{amsfonts}
\usepackage{amsmath}
\usepackage{graphicx}
\usepackage[font=small]{caption}
\usepackage{subcaption}
\usepackage{bm}
\usepackage[ruled,noend]{algorithm2e}
\usepackage{algpseudocode}
\usepackage{varwidth}
\usepackage{placeins}
\usepackage[utf8]{inputenc}
\usepackage{mathtools} 
\usepackage{hyperref}
\usepackage{booktabs}
\usepackage{verbatim}
\usepackage{enumitem}
\usepackage{soul}
\usepackage{xcolor}
\newcommand{\footremember}[2]{%
    \footnote{#2}
    \newcounter{#1}
    \setcounter{#1}{\value{footnote}}%
}

\newcommand{\qed}{\rule{1.5mm}{2mm}\vspace{0.1in}}
\newenvironment{proof}{\par\noindent{\bf Proof:}}{\qed}

\newcommand{\ignore}[1]{}

\newtheorem{theorem}{Theorem}
\newtheorem{lemma}{Lemma}
\newtheorem{corollary}{Corollary}
\newtheorem{proposition}{Proposition}
\newtheorem{observation}{Observation}

\newtheorem{definition}{Definition}

\algtext*{EndFor}
\algtext*{EndIf}
\DeclareMathOperator*{\argmin}{argmin} 

\title{Optimal Collaterals in Multi-Enterprise Investment Networks\thanks{This work has been published in Proceedings of the ACM Web Conference 2022 (WWW’22). \href{https://doi.org/10.1145/3485447.3512053}{https://doi.org/10.1145/3485447.3512053.}}} 

\author{
			 Moshe Babaioff\footremember{MS}{Microsoft Research, Israel. Email: moshe@microsoft.com.}%
  \and Yoav Kolumbus\footremember{HUJI}{The Rachel and Selim Benin School of Computer Science and Engineering, The Hebrew University of Jerusalem, Israel. Email: yoav.kolumbus@mail.huji.ac.il.}%
  \and Eyal Winter\footremember{LANCASTER}{Department of Economics, Lancaster University,  Lancashire, England. The Center for the Study of Rationality and Department of Economics, The Hebrew University of Jerusalem, Israel. Email: eyal.winter@huji.ac.il.}
	}
\date{}

\begin{document}
\maketitle

\begin{abstract}
We study a market of investments on networks, where each agent (vertex) can invest in any enterprise linked to her, and at the same time, raise capital for her firm's enterprise from other agents she is linked to. Failing to raise sufficient capital results with the firm defaulting, being unable to invest in others. Our main objective is to examine the role of collateral contracts in handling the strategic risk that can propagate to a systemic risk throughout the network in a cascade of defaults. We take a mechanism-design approach and solve for the optimal scheme of collateral contracts that capital raisers offer their investors. These contracts aim at sustaining the efficient level of investment as a unique Nash equilibrium, while minimizing the total collateral.

Our main results contrast the network environment with its non-network counterpart (where the sets of investors and capital raisers are disjoint). We show that for \emph{acyclic} investment networks, the network environment does not necessitate any additional collaterals, and systemic risk can be fully handled by optimal bilateral collateral contracts between capital raisers and their investors. This is, unfortunately, not the case for \emph{cyclic} investment networks. We show that bilateral contracting will not suffice to resolve systemic risk, and the market will need an external entity to design a global collateral scheme for all capital raisers. Furthermore, the minimum total collateral that will sustain the efficient level of investment as a unique equilibrium may be arbitrarily higher, even in simple cyclic investment networks, compared with its corresponding non-network environment. Additionally, we prove computational-complexity results, both for a single enterprise and for networks.
\end{abstract}

\bibliographystyle{splncs03}
\section{Introduction}
Firms and investment entities that raise capital to undertake profitable investment enterprises often face strategic risk. 
Investors may decline to invest because of concerns that the firm will not manage to raise sufficient capital due to other investors declining to invest. This might lock the investment game in an inefficient 
equilibrium of mis-coordination. When these firms are also connected with one another in a network of investments, such strategic risk can propagate throughout the network and may yield consequences to the entire market. The purpose of this paper is to study such network environments, reveal the systemic strategic risk that they might be exposed to, and show how the negative consequences of such risks can be remedied through optimal collateral schemes. 

Unlike standard risk that can rise from a variety of forces on which we have little control, strategic risk arises due to multiplicity of equilibria and can often be alleviated by means of incentives. In our context these incentives will take the form of collaterals that will guarantee sufficient investment and hence will not be paid on the equilibrium path. In the framework of a single firm raising capital, investors' strategic considerations 
only involve contemplating about the decisions of other potential co-investors in that underlying enterprise, 
and the consequences of their decisions. 
In the framework of investment networks, the scope of players' considerations is broader. 
An investing firm should not only consider the incentives of a co-investor to invest, but also the co-investor's potential insolvency, which, in turn, depends on the incentives of other firms to invest in this co-investor's enterprise. 

Our formal model of investments in a \emph{single enterprise} assumes that the capital-raising firm has some enterprise  
that promises a certain return rate if the capital raised is sufficiently high to cover the enterprise operational costs.  
There is a set of potential investors, each has an opportunity to invest some investor-specific amount in the enterprise, and needs to decide if she invests or not. If the operational costs are not covered by the investments made, the enterprise is in default, and investors' returns are zero. 
If there is no default then the enterprise's returns are divided among the investors proportionally to their investments. Hence the firm's liability for these investments is only limited to the outcome of its enterprise, even when it has other funds in its possession.   
To remove strategic uncertainty that might make investors reluctant to invest, the firm offers collateral contracts to its investors, bounding the loss arising from strategic risk. The objective of the firm is to find an optimal (minimum total) collateral scheme guaranteeing that all investors choose to invest (and hence default does not occur)  in a unique equilibrium of the investment game. 

Our model of the \emph{network investment game} extends the single-enterprise model and has the fundamental additional problem of default cascades.
Unlike in the single-enterprise case, in an investment network each investor may have investment opportunities in multiple enterprises, and for each of them, she has to decide whether to invest or not. Moreover, some of these investors are actually raising capital for their own enterprises, and may default if they do not raise sufficient capital. When a firm defaults it immediately withdraws all its investments, and we assume a worst-case scenario of zero recovery for the defaulted firm,  as common in the literature (see, e.g., \cite{caccioli2018network,gai2010contagion}).    
We assume that the network is exogenously given, representing existing business relationships between the firms. 
The crucial difference from the single-enterprise case is that a firm that has defaulted cannot make investments in other enterprises, which might create a cascade of defaults. Hence, in a network environment each investor must, in principle, assess not only the incentives of other co-investors to invest but also the possibility that some of these co-investors may default and will fail to invest.
The network collateral problem seeks to find an optimal (minimum total) collateral scheme guaranteeing that for the sub-network of efficient enterprises, all investors choose to invest in all enterprises (and hence default does not occur) in a unique equilibrium of the network investment game. Importantly, these collaterals must be sufficient to overcome the systemic risk of default cascades. 
See Section \ref{sec:model-all} for the full details of the model and for the formal definition of the optimization problem.

Our analysis starts by characterizing these optimal schemes, i.e., those schemes that minimize the overall collateral given to sustain the desired outcome of efficient investments in a unique equilibrium. 
We show that whenever full investment can be established as unique equilibrium by 
collateral contracts, an optimal scheme generates a dominance solvable game. 
Namely, the unique equilibrium is obtained by 
iterative elimination of strictly dominated strategies (see further details in Section \ref{sec:characterization}). 

In Section \ref{sec:single} we focus on single-enterprise networks. We characterize the optimal schemes, showing that in such networks 
they must exist. Additionally, we show that the problem of finding an optimal scheme for an enterprise with $n$ investors is NP-hard. 

Then, we return in Section \ref{sec:investment-networks} to the general network game. 
In studying network environments in terms of their susceptibility to systemic strategic risk our results highlight an important distinction between cyclic and acyclic networks. When the network is \emph{acyclic} we show that there exists an optimal collateral scheme for the entire market that can be represented as a collection of smaller schemes, one for each enterprise, such that the collaterals that each firm is promising its potential investors are identical to the ones that would have been made under the assumption that these investors have no external obligations and hence all the investors are capable to invest if provided the incentives to do so. 
Put differently, a firm designs its collateral scheme as if the entire market consists of itself and its group of potential investors, ignoring all the rest. 
We show how to calculate the optimal scheme and the order of iterated elimination of dominated strategies of acyclic networks. 
The order of iteration depends on the structure of the network, and in acyclic networks it is induced by a topological order on the firms with an enterprise. 
In particular, the iteration starts from investments of firms that do not hold their own enterprises, and
continues in a way that ensures that defaults are no longer a concern when each decision is taken.     
From the computational perspective, in acyclic networks with $n$ firms in which every enterprise has at most $d$ potential investors, 
we show that this solution can be obtained in $\mathcal{O}(2^d dn)$ time. In particular, it is polynomial when $d$ is a constant.   

Unfortunately, under \emph{cyclic} networks an overall optimal collateral scheme may {\em not} be reduced to a collection of optimal schemes -- one for each enterprise. Indeed we show that for some (simple) cyclic networks, the total collateral needed to prevent systemic strategic risk is arbitrarily higher than the total collateral firms would provide when considering only their own investors. The implication of this finding is that under cyclic networks the prevention of systemic strategic risk may not be decentralized, in the sense that it would require an outside entity to resolve the potential mis-coordination. 
Furthermore, there are cyclic networks in which full investment is an equilibrium, but it cannot be established as unique equilibrium with any collateral scheme. 
These results follow from the fact that providing a collateral to potential investors, ignoring the possibility of defaults, is not a sufficient instrument to make sure that all of them would invest. Finally, we show that cycles introduce further difficulty on the computational side: we show that finding an optimal collateral scheme in cyclic networks is NP-hard, even if a solution for each individual-firm problem is known. We prove this by presenting a reduction from the minimum feedback vertex set problem. For further details see Section \ref{sub-sec:general-graphs}.

\vspace{3pt}
\noindent
\textbf{Discussion of the literature:} 
In the case of a single enterprise our model is similar to \cite{Halac-et-al-2020}. The main difference between the models lies in the fact that here we secure investments through collaterals that are not paid on the equilibrium path instead of offering higher interest rates on investments. Another difference is in the causes for default. In \cite{Halac-et-al-2020} default is stochastic and is triggered when the total capital raised is below a random threshold. Here instead default occurs when the operational costs exceed the total invested capital. 
Our somewhat simpler single-enterprise model has been introduced 
in order to allow us to focus on investment networks, rather than only a single-enterprise problem.


A different single-enterprise investment game is studied in \cite{Ban-and-Koren-EC2020}. 
Our model and analysis differs in basic features. Most importantly, we study network environments with heterogeneous investors and consider a mechanism-design problem for these networks.


Systemic risk is widely studied in the context of inter-bank liability networks. We refer 
to two surveys on this literature \cite{caccioli2018network,jackson2020systemic} 
and mention below several works that are closer to our context.  
The survey \cite{jackson2020systemic}  
classifies the sources of systemic risk to two types: direct externalities,  
and self fulfilling feedback effects in systems with multiple equilibria. 
Our work is conceptually related to the later type, however, the equilibria usually studied in the systemic-risk literature are different than the concept that we consider. In liability-clearing models, an equilibrium relates to a fixed point of the clearing process. In this process, network vertices (banks) do not make decisions. In contrast, in our setting, the network of investment opportunities induces a game of the investors, and an equilibrium is a Nash equilibrium of the game. The problem we study is that of the design of incentives in the complex multi-player game induced by a networked environment of investment opportunities, so as to induce desired equilibrium outcomes. In this sense, our approach is more closely related to the fields of mechanism design and algorithmic mechanism design \cite{nisan2001algorithmic}, and specifically, to works in the mechanism-design literature that study models of minimal intervention for obtaining desired outcomes in games (e.g., \cite{tennenholtz-k-implementation-2004,tennenholtz-2009}).

In \cite{schuldenzucker_portfolio_EC2020}, the authors study the effect of eliminating network cycles via portfolio compression in banking networks 
under the model of \cite{rogers2013failure}. 
Their discussion of incentives is conceptually close to our approach, but the problems studied are different. 
The problem in \cite{schuldenzucker_portfolio_EC2020} considers systemic risk due to external shocks, in the spirit of \cite{eisenberg2001systemic,rogers2013failure}, while we study a strategic setting in which investors make decisions, and we focus on systemic risk resulting from strategic risk. 
Another difference between banking liability networks and the investment setting is that 
firms' liability to investors in our setting is more limited than that of banks.
This makes the investment network more susceptible to systemic risk, 
because milder obligation to compensate investors in case they have losses would mean less resources for these investors to return their own debts. 

The role of cycles in clearing problems is also discussed in \cite{jackson2019distorted}. 
Their model describes banking networks that in addition to liabilities include also equity connections, thus generalizing the models of \cite{eisenberg2001systemic,gai2010contagion,elliott2014financial}. The authors show that the existence of multiple solutions to the liability-equity clearing problem in their model is related to cyclic debt relations. Given that all investments have already been made, they use bailouts in retrospect to fix default cascades due to shocks. 
In contrast to the clearing problem of \cite{jackson2019distorted}, we study a game in which individuals make investment decisions, and the collaterals are designed to affect these investment decisions.

In \cite{schuldenzucker2020default} the liability-clearing problem of \cite{eisenberg2001systemic,rogers2013failure} is extended to discuss banking networks which also include derivative links in the form of credit default swaps (CDS), which have some similarity to collaterals. In a CDS an investor, or lender, buys an insurance-like contract from a third party seller for the case that the receiver of the investment defaults. In contrast to CDS contracts that are bought by investors, in our case a collateral is offered by the capital-raising firm to its potential investors to incentivize them to engage in the investment in the first place. Moreover, typically, a CDS is a more limited form of security than a collateral, as the execution of the contract depends on the solvency of the third party who sells the CDS. Indeed, in the context of systemic risk the triggering of a CDS contract due to a default may be correlated with a default of the CDS seller itself.  

An additional related line of work on inter-bank liability networks discusses endogenous network formation models \cite{babus2017endogenous,erol2018network,acemoglu2015systemic,babus2013formation,farboodi2014intermediation}. 
Key differences from our work are that we focus on strategic risk and on the systemic risk resulting from it, while these works consider systemic risk due to external shocks, and importantly, we study a mechanism design problem where collateral contracts are used for inducing an efficient equilibrium as a unique outcome. 

Collateral contracts are among the oldest and most basic forms of a financial securities.
In two classic works, \cite{bernanke1986agency,bernanke1990financial}, Bernanke and Gertler highlighted the role of collaterals in investment markets,  
showing how high collateral amounts induce higher investments and a decrease in collateral may trigger a decrease in investments. 
In their model the mechanism is that of agency costs for investors due to incomplete information; these costs increase when the firms raising capital have less collateral. In our model collateral contracts secure an explicit collateral per investor, and so a similar relation between collateral and investment exists even in a full information setting, due to strategic considerations of rational players. 

To our knowledge our work is the first to study collateral contracts as a mechanism-design tool in strategic settings for incentivizing investment and handling systemic strategic risk.

\section{Model}\label{sec:model-all}
\subsection{The Investment Game}\label{sec:model-invest}
Consider a network of 
investments with $n$ vertices where each vertex represents a firm. 
A firm may have an ``investment enterprise'' and it may also act as an ``investor agent'' that considers investments in enterprises of other firms in the network. 
Each investment firm which has an investment enterprise is connected with outgoing edges to a set of investor agents. 
A directed edge from vertex $k$ to vertex $i$ with weight $x_{ki} > 0$ represents an opportunity of agent $i$ to invest an amount $x_{ki}$ in the enterprise of $k$.
Investor agents decide whether to accept or decline their investments opportunities (for the entire amount $x_{ki}$).\footnote{
This modeling is similar to \cite{segal1999contracting,Halac-et-al-2020} 
and is meant to capture the fact that many investment campaigns involve the capital-raising firm 
fixing the desired investment from each investor in advance, 
depending on its financial needs and on investor constraints, 
as well as on the firm's desired allocation of control between investors.} 
A firm that has multiple incoming edges representing multiple investment opportunities that the firm considers makes this binary decision for each of its investments, and gains utility that equals the sum of utilities from all its investments. Every firm that has an investment enterprise has some operational cost, and if the firm does not raise sufficient capital to cover that cost, the firm defaults. 
A defaulting firm immediately withdraws all of its investments, with zero recovery, so it essentially makes no investments in any enterprise.
The network structure is assumed to be an exogenous input, representing a set of investment opportunities at a given time (say, at the beginning of the year).\footnote{The issue of endogenizing the network formation is an interesting one, but goes beyond the scope of the current paper.} 
To convince investors to invest, firms may offer collaterals to their potential investors, to reduce the strategic risk for these investors.
Next we give the formal definition of the model and of the game induced by the network of investments.

First let us define the basic components: 
There is a set $V$ of $n$ firms, with $[n]$ denoting firm indices. 
There is a graph $G = (V,E)$ describing the topology of investment interactions, with an edge $(k,i)\in E$ indicating that $i$ has an investment opportunity in $k$. Each edge $(k,i)\in E$ has a weight $x_{ki}$, specifying the amount $i$ can invest in $k$. For convenience we assume the investment opportunities are given by a matrix $\mathcal{X}$, which is 0 for every non-edge.\footnote{Although $\mathcal{X}$ provides the information of the topology of $G$, the separation in notation between $G$ and its weights $\mathcal{X}$ is useful for a discussion of network topology later on.}
Let $\mathcal{Z} = \{Z_i\}_{i=1}^n$, $\mathcal{A} = \{\alpha_i\}_{i=1}^n$ with $Z_i, \alpha_i \geq 0$ be enterprise $i\in V$ cost and interest rate parameters, as specified below.
Additionally, denote by $V_e=\{v \in V: \exists u \in V \ s.t. \ (v,u) \in E \}$ the set of \emph{enterprise vertices of $G$}, i.e., the subset of vertices that have potential investors (have a non-zero out-degree in $G$). For every vertex $i\notin V_e$, both $\alpha_i$ and $Z_i$ are irrelevant and can w.l.o.g. be set to $0$.
The model is as follows.  
\begin{itemize}
\item Let $X_k$ denote the sum of investment opportunities in enterprise $k$: $X_k = \sum_{i} x_{ki}$.
\item Each firm $i$ with an \emph{investment opportunity $x_{ki}>0$} in $k$ ($(k,i) \in E$) chooses one of two actions: (i) invest the full amount $x_{ki}$ in $k$ and wait to receive returns at the end of the year (\emph{cooperate} w.r.t. $k$); or (ii) decline the investment opportunity in $k$ (\emph{defect} w.r.t. $k$). Notice that agents with multiple investment opportunities (incoming edges) make such a decision for each one separately.
\item The operational cost of enterprise $k$ is captured by the cost parameter $Z_k\geq 0$.   
\item Denote the set of defectors w.r.t. $k$ by $D_k$ and the set of cooperators w.r.t. $k$ by $C_k$. An edge $(k,i)$ is called a \emph{cooperate edge} if $i \in C_k$ and a \emph{defect edge} if $i \in D_k$.
\item At the beginning of the year, each defector $i \in D_k$ keeps the full amount $x_{ki}$.\footnote{To simplify the analysis we assume no interest is given on funds not invested. Adding such an interest will not change the model significantly, and our results can be easily adjusted appropriately.}
Firms may reach a state of default if they do not raise sufficient investments to cover the operational costs of their enterprises.
Denote by $I_k$ the set of players that are in $C_k$ and not in default, these players \emph{invest} in enterprise $k$.
If the total investment in $k$ does not cover the operational cost ($\sum_{i \in I_k} x_{ki} < Z_k$), then the enterprise cannot be completed and we say that $k$ has \emph{defaulted}.
If firm $k$ has defaulted then it withdraws all its investments with zero recovery. 
The set of investors $I_k$ is induced by the strategy profile played by all players (as specified by Definition \ref{default-determination} below). 
An edge $(k,i)\in E$ is called an \emph{invest edge} if $i \in I_k$. 
\item If $k\in V_e$ is not in default, the total net invested amount (after deducting the operational cost) 
accrues interest of $\alpha_k > 0$; that is, the total amount to be distributed among the investors in $k$ is $(1+\alpha_k)(\sum_{i\in I_k} x_{ki} - Z_k)$. 
The \emph{enterprise return} $R_{ki}$ from $k$ to an investor $i\in I_k$ is then $i$'s proportional share of the total return from $k$: 
	{\small
	\begin{equation}\label{R_i(II)}
	R_{ki} = 
	\max \bigg(0, 
	\underbrace{
	(1+\alpha_k)\bigg(\sum_{\substack{j \in I_k}}x_{kj} - Z_k\bigg)
	}_{\text{Total net return}}
	\hspace{-8pt}
	\underbrace{
	\frac{x_{ki}}{\sum_{\substack{j \in I_k \vspace{5pt}}}x_{kj }} 
	}_{\text{$i$'s proportional share}}
	\hspace{-8pt}
	\bigg)
	\end{equation}
	}	
\item Let $c_{ki}\geq 0$ be a collateral given to player $i$ for the investment $x_{ki}$. If the enterprise return $R_{ki}$ for player $i \in I_k$ is lower than the investment $x_{ki}$, then the collateral $c_{ki}$ can be realized so that the utility of player $i$ from $k$ 
is $R_{ki} + c_{ki}$, capped at $x_{ki}$, i.e., $\min (R_{ki} + c_{ki}, x_{ki})$.
In case player $i$ has profits (i.e., $R_{ki} > x_{ki}$) the collateral is not realized.
Thus, the utility from $k$ for a player that defects ($i \in D_k$) is $x_{ki}$,
her utility if she cooperates yet in default ($i\in C_k\setminus I_k$) is $0$,
and her utility if she invests ($i \in I_k$) with collateral $c_{ki}$ is:
{\small
\begin{equation}\label{eq:cooperate-utility}
		\begin{array}{lr}
			U_{ki} = \left\{\begin{array}{lr}
						c_{ki} + R_{ki}, & \text{ } R_{ki} \leq x_{ki} - c_{ki}\\
						x_{ki},           & \text{ } x_{ki}-c_{ki}\ < R_{ki} \leq x_{ki}\\
						R_{ki},       & \text{ } x_{ki} < R_{ki} \end{array}\right.
		 \end{array}
	\end{equation}
}
The total utility of each player $i$ is $U_i = \sum_{k\in [n]} U_{ki}$, where $U_{ki} = 0$ if $x_{ki} = 0$ 
(that is, if $(k,i)\notin E$).
\end{itemize}

\noindent
Given a set $C$ of \emph{cooperate edges}, the set of firms in default and the induced set of \emph{invest edges} are determined in the following process: 
\begin{definition}\textbf{(default determination)}\label{default-determination}:
Initialize $T=\emptyset$ as the set of firms in default. Assume that $I=C$ are the tentative invest edges. (i) For each firm $k\in V_e$, if it is in default when the set of investments is $I$, add the firm to the set $T$ and remove from $I$ all edges pointing to this firm; (ii) Repeat until no more firms are added to $T$.  
\end{definition}

Finally, Let $\textbf{c}$ be the matrix of collaterals given to the players for each of their investments, where $c_{ki}$ is the collateral given to player $i$ for an investment of an amount $x_{ki}$ in enterprise $k$, and if $x_{ki} = 0$ then $c_{ki}$ is defined to be zero. 
The \emph{collateral matrix} $\textbf{c}$ together with the inputs $\left\{n, \mathcal{X}, \mathcal{Z}, \mathcal{A} \right\}$ induce a simultaneous-move full-information game with the $n$ investors as players.

\subsection{The Mechanism-Design Problem}\label{sec:design}
In this section we define the mechanism design problem of finding optimal collaterals to induce the efficient level of corporation as a unique equilibrium (when it is possible) in the investment network model described in the previous section. We also provide some definitions that will be useful later. 
 
We assume that ties are broken towards investing, that is, if given fixed actions of the other players a player that has equal utility from investing 
(playing \emph{cooperate} and not defaulting) and from defecting, 
then the player strictly prefers to invest.\footnote{Note that defaulting results with zero utility thus is always dominated by defecting.} 
This simplifies the writing as players can never be indifferent between the two actions. Formally this tie-breaking rule is equivalent to adding a fixed infinitesimal utility to  cooperating.\footnote{A different interpretation is that an equality in utilities from defecting and cooperating occurs when the worst-case utility of cooperating after an iterated elimination of strictly dominated strategies equals the best-case utility of defecting, and so players may be slightly optimistic regarding the actions of their peers -- as long as this optimism does not carry any risk.} 

Clearly an enterprise should not take place if, even assuming that all investors cooperate, the return is not sufficient to distribute to investors their original invested amount.
Thus, it is efficient that only enterprises that are profitable will take place (and others do not), and so we 
focus on the subset of profitable enterprises, ignoring unprofitable enterprises (by removing them from the network). Thus, in the rest of the paper we assume that all investment enterprises are profitable when all players cooperate: {\small$\forall k\in V_e, \ (1+\alpha_k)(X_k - Z_k) \geq X_k$}. 
It is easy to see that in such a case, everyone cooperating is a Nash equilibrium, and no collaterals are realized if all players cooperate. 
 
\begin{observation} \textbf{(profitable investments)}:\label{cooperate-equilibrium}
	Assume that {\small $\forall k\in V_e, \ (1+\alpha_k)(X_k - Z_k) \geq X_k$} (all enterprises are potentially profitable). Then for every collateral matrix the strategy profile in which all players cooperate is a Nash equilibrium, and no collateral is realized in equilibrium.
\end{observation}

Yet, this Nash equilibrium might not be unique, and all agents defecting might also be an equilibrium.\footnote{E.g., when the cost of an enterprise exceeds the amount of any single investment opportunity in that enterprise.} 
To avoid this bad equilibrium, we are interested in finding collaterals that will induce all-cooperation as the \emph{unique} equilibrium, whenever possible.

We focus on collateral matrices that induce all-cooperating as a unique Nash equilibrium, 
we call such matrices \emph{viable collateral matrices}.  
\begin{definition} \textbf{(viable collateral matrix)}: \label{viable-collaterals}
The collateral matrix $\textbf{c}$ is called a \emph{viable collateral matrix} if the game induced by $\textbf{c}$ with inputs  
$\left\{n, \mathcal{X}, \mathcal{Z}, \mathcal{A} \right\}$  has a \emph{unique} Nash equilibrium where all the players cooperate.
\end{definition}

In this paper we consider the  problem of  finding a viable collateral matrix that has minimal total cost, when exists. 
\begin{definition} \textbf{(the collateral-minimization problem)}: \label{min-collaterals-problem}
Let \textbf{$\mathcal{C}$} denote the set of all viable collateral matrices given the input $\left\{n, \mathcal{X}, \mathcal{Z}, \mathcal{A} \right\}$. 
The \emph{collateral minimization problem} asks to find a viable collateral matrix, if such exists,  
with minimum sum of collaterals (1-norm):
{\small
\begin{equation}\label{eq:min-collaterals-problem}
	\textbf{c}^* \in  \argmin_{\textbf{c} \in \textbf{$\mathcal{C}$}} \sum_{i,j} c_{ij} 
\end{equation}
}

\noindent
We call the collaterals of such a solution $\textbf{c}^*$ \emph{optimal collaterals}. Inputs for which no viable collateral matrices exist can be identified in polynomial time, see Proposition \ref{thm:unique-equilibrium-existence}.
\end{definition}
Note that since  $\mathcal{C}$ is an infinite set, Equation (\ref{eq:min-collaterals-problem}) should be written as an infimum, yet we write it as minimum instead of infimum
since our tie-breaking rule ensures that the infimum is obtained. 

To solve the collateral minimization problem it will be useful to consider the following matrices which 
for any investment opportunity promise the minimal collateral that is necessary 
 to incentivize the player to cooperate, given that all other collaterals are fixed.

\begin{definition} \textbf{(minimal collateral matrix)}: \label{minimal-collaterals}
	The collateral matrix $\textbf{c}$ is called a \emph{minimal collateral matrix} if: (i) $\textbf{c}$ is a viable collateral matrix, and, (ii) $\forall i$, $\forall \varepsilon > 0$ the matrix $\textbf{c} - \varepsilon \hat{e}_{ij}$ is not a viable collateral matrix, where $\hat{e}_{ij}$ is the $n \times n$ matrix with all entries equal zero except the $i,j$'th entry which equals one.  
\end{definition}
Notice that a minimum norm \emph{viable} collateral matrix must also be a \emph{minimal collateral matrix} and hence a solution for the collateral minimization problem is a minimal collateral matrix that minimizes the sum of collaterals. 

The following definitions distinguish between two types of collaterals, which will be useful later on for characterizing the optimal collaterals. Each collateral can be either a ``full collateral'', covering the full investment of a player, or a ``partial collateral'' which guarantees a minimum payoff that is lower than the investment. 
\begin{definition} \textbf{(full collateral)}: \label{full-collateral}
	For investment opportunity $x_{ki} > 0$,  collateral $c_{ki}$ is called \emph{full collateral} if $c_{ki} = x_{ki}$.
\end{definition}
\begin{definition} \textbf{(partial collateral)}: \label{partial-collateral}
	For investment opportunity $x_{ki} > 0$, collateral $c_{ki}$ is called a \emph{partial collateral} if it is not a full collateral (notice that this includes zero collaterals).
\end{definition}
%
\section{Characterization}\label{sec:characterization}

Here we characterize several basic properties of the collateral minimization problem.
We first show that 
existence of a unique equilibrium implies dominance solvability. We then present a polynomial time procedure to find for a given network whether a unique equilibrium can be induced by collateral contracts. Finally,  
we characterize the conditions under which a player will prefer to cooperate even with a zero collateral, and the conditions under which a player will prefer to cooperate only when provided with a full collateral. 
%
\subsection{Monotonicity and Dominance Solvability}
An important property of investments and of their associated strategic risk, is that investment gains increase as additional potential investors decide to invest, and decrease as cooperating investors stop investing. More specifically, for any player that invests in an enterprise for a given cooperation profile on all other investments opportunities, her gains can only increase as decisions are turned from \emph{defect} to \emph{cooperate}. This is shown in the following proposition.  
\begin{proposition}\label{thm:monotonicity} \textbf{(monotonicity)}:
\begin{enumerate}
\item For any investment opportunity $x_{ki}$ of agent $i$ in enterprise $k$, the enterprise return to agent $i$ is a monotone increasing function 
of the set of investment opportunities in the network in which players cooperate.
I.e., if $R_{ki}(C)$ denotes the enterprise return of $i$ that invests in enterprise $k$ ($i\in I_k$) when $C$ is the set of cooperate edges, then for every set $B\supset A$ it holds that $R_{ki}(B) \geq R_{ki}(A)$.
	
\item Fix any subset $A$ of edges, any enterprise $k\in V_e$, and any player $i$ such that $(k,i) \in E$ and $(k,i) \notin A$. Assume that when the set of cooperate edges is $A$ then player $i$ strictly prefers to cooperate over defecting w.r.t. $k$. Then, whenever the set of cooperate edges is $B\supset A$ player $i$ strictly prefers to cooperate over defecting w.r.t. $k$ as well.
\end{enumerate}
\end{proposition}
\begin{proof}
First, notice that a superset of cooperate edges in the entire network induces a superset of cooperate edges from every individual enterprise vertex.
Second, a superset of cooperate edges induces a superset of invest edges, since additional cooperate actions cannot 
turn an enterprise that did not default to a defaulting enterprise. 
Hence the set of investors in enterprise $k$ when $B$ are the cooperate edges, $I_k(B)$, is a superset of the set of investors in $k$ when $A\subset B$ are the cooperate edges, $I_k(A)$.
For a given set of investors $I_k$ in enterprise $k$ let $g(I_k) = \sum_{\substack{j \in I_k}}x_{kj}$. 
The enterprise return of a player $i \in I_k$, 
as given by Equation (\ref{R_i(II)}), 
is a weakly increasing function of $g(I_k)$. 
Since $\forall k,j, \ x_{kj} \geq 0$ and $I_k(A) \subseteq I_k(B)$, it follows that $g(I_k(B)) \geq g(I_k(A))$, and hence $R_{ki}(B) \geq R_{ki}(A)$, which concludes the proof of $(1)$.
Under the second condition, the utility of player $i$ from enterprise $k$ if she decides to cooperate, as given by Equation (\ref{eq:cooperate-utility}), is a weakly increasing function of $R_{ki}$ and hence the utility of player $i$ from cooperating can only increase as more players cooperate. The utility of the alternative action (\emph{defect}) is constant. Thus, if \emph{cooperate} is a best reply of $i$ when $A$ is the set of cooperate edges, it is also a best replay when $B \supset A$ is the set of cooperators. Notice that under condition $(2)$ player $i$ is not in default, or otherwise she would prefer to defect.  
\end{proof}

A direct corollary is that monotonicity also holds in the inverse direction: If \emph{defect} is a best reply of a player when the players cooperate in some set of other investment opportunities, then \emph{defect} is a best reply also when players cooperate only in a subset of these opportunities.

The next result shows that for a collateral matrix that induces a game with a unique Nash equilibrium in which all players cooperate, the induced game is dominance solvable, and so this equilibrium can be reached by a series of iterated eliminations of dominated strategies (see Definition \ref{dominance-solvable game}). 
I.e., there is at least one player with a strictly dominant strategy to cooperate, at least one player that strictly prefers cooperating given the dominant strategy of the first player, and so on, until the strategy profile of \emph{all-cooperate} is reached. 
For completeness, we start with the definition of a dominance solvable game.  
\begin{definition} \textbf{(dominance solvable game)}: \label{dominance-solvable game}
A game is called \emph{dominance-solvable} if there exists an order $\sigma$ over the strategies of the players of iterated elimination of strictly dominated strategies which leads to a single strategy profile. 
\end{definition}

\begin{proposition}\label{thm:uniqueness} \textbf{(unique equilibrium condition)}:
If the game induced by the collateral matrix $\textbf{c}$ with input $\left\{n, \mathcal{X}, \mathcal{Z}, \mathcal{A}\right\}$ has a unique Nash equilibrium in which all players cooperate in all their investment opportunities then the game is dominance solvable. 
\end{proposition}
\begin{proof} 
Assume that \emph{all-cooperate} is the unique Nash equilibrium and consider the strategy profile where all players play \emph{defect} in all their investment opportunities (incoming edges). 
Then, there must be an edge, which we will denote by $\sigma_1$, corresponding to a decision of a player in which the player strictly prefers to cooperate,  
or otherwise \emph{all-defect} would also be a Nash equilibrium in contradiction to our assumption. 
Notice that this player cannot be in default for this strategy profile, or otherwise she would prefer to defect. 
After eliminating the \emph{defect} strategy for $\sigma_1$, there must be an edge $\sigma_2$ in which the player strictly prefers to cooperate, otherwise the strategy profile of \emph{all-defect} except $\sigma_1$ would be a Nash equilibrium. 
Due to Proposition \ref{thm:monotonicity}, if the player of edge $\sigma_2$ cooperates, the player of $\sigma_1$ still strictly prefers to cooperate.
After eliminating the \emph{defect} strategy  for $\sigma_2$ we can similarly eliminate the \emph{defect} strategy for a third edge, $\sigma_3$, and again, by  Proposition \ref{thm:monotonicity}, the players of the previous edges, $\sigma_1$ and $\sigma_2$, will still strictly prefer to cooperate when there are more cooperate edges. 
The same argument holds until the last edge. Thus, if \emph{all-cooperate} is the unique Nash equilibrium, due to monotonicity as stated in Proposition \ref{thm:monotonicity}, there exists an ordering of the edges $\sigma_1,...,\sigma_{|E|}$ of iterated elimination of strictly dominated strategies that leads to the \emph{all-cooperate} strategy profile. 
\end{proof}

\subsection{Unique Equilibrium through Collaterals}\label{sec:unique-eq-through-collaterals}
The interest of the collateral minimization problem is in finding optimal collaterals to induce a unique equilibrium of full investment; this rather strong solution concept as a mechanism-design objective is aimed to assure that there are no 
mis-coordination issues between the players and that the game reaches an efficient outcome. Seeking a unique equilibrium is well motivated
by empirical literature showing how players often converge to inefficient equilibria in the presence of strategic risk (see e.g., \cite{Brandts-and-Cooper-AER-2006}).
Interestingly, unlike the non-network case (where each enterprise has its own set of investors), there are network structures in which a unique equilibrium cannot be induced by collaterals. These cases occur in cyclic investment networks, and are due to a special type of cyclic structures in the network. 
Our characterization shows how to identify 
such cases and to isolate these problematic structures. 
Specifically, in single-enterprise networks as well as in all directed acyclic networks, 
there are no problematic structures and a unique equilibrium can always be induced by collaterals. 
The following result, Proposition \ref{thm:unique-equilibrium-existence}, presents a constructive characterization of investment networks in which a unique equilibrium can be induced by means of collaterals. This is done by an iterative process that runs in polynomial time. 
\begin{proposition}\textbf{(solvability by collateral contracts)}:\label{thm:unique-equilibrium-existence}
Consider the following iterative process: If there exists an enterprise $k$ in the graph such that the total investments from non-enterprise vertices (spikes) in this enterprise is at least the cost of that enterprise ($Z_k$), remove this enterprise and its non-enterprise investor vertices and incident edges 
and replace every investment opportunity that $k$ has by the same opportunity from a new (spike) vertex. Then, reiterate on the new graph. 
It is possible to induce a unique equilibrium by collaterals if and only if the above process ends up with the empty graph.
\end{proposition}

Figure \ref{fig:iterative-process} presents an example illustrating one step of the iterative process. The proof of the proposition as well as further details can be found in Appendix \ref{appendix:existence-of-viable-collateral-matrices}.

\begin{figure}[t]
	\centering
		\includegraphics[width=0.55\textwidth]{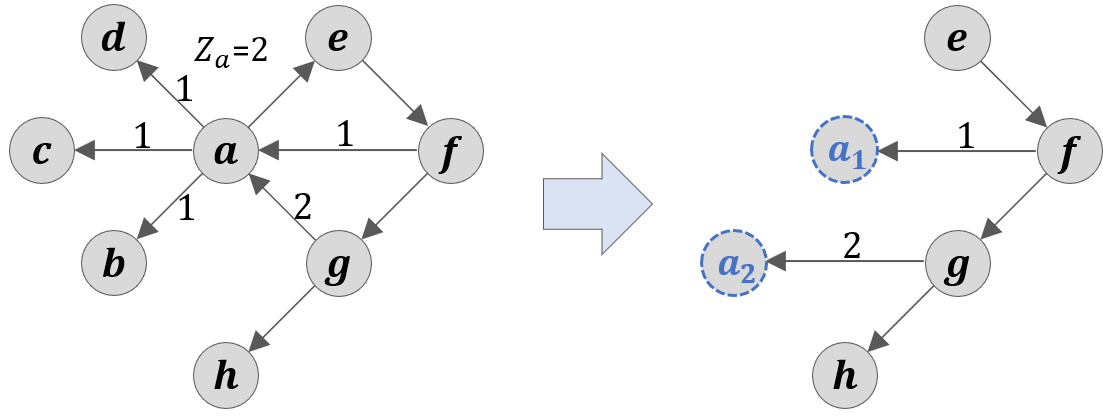}
	\caption{{\small
	An example of a step in the iterative process of Proposition \ref{thm:unique-equilibrium-existence}. Letters indicate vertex names and numbers near directed edges indicate the investment opportunity amounts. E.g., vertex $a$ has an investment opportunity of amount $2$ in the enterprise of vertex $g$. Edges where the numbers are not indicated can have any positive value in this example. Left: The network before the iteration step. Vertex $a$ is an enterprise vertex with operational cost $Z_a=2$. Vertex $a$ is also a cycle vertex in the network, and has three non-enterprise investors (spike vertices), each with an investment amount of $1$, totaling more than the cost $Z_a$, and thus (full) collaterals can ensure $a$ will not default. Right: The new network after the iteration step, capturing the ensured solvency of $a$. Vertex $a$ is removed together with it's spike vertices ($b$, $c$, $d$)  and incident edges. The investment opportunities of vertex $a$ of amounts $1$ and $2$ are replaced by investments of the same amounts from new spike vertices, indicated in the figure with dashed boundaries and denoted by $a_1$ and $a_2$.  
	}}
	\label{fig:iterative-process}
\end{figure}

\subsection{Zero and Full Collaterals}
In Section \ref{sec:design} we partitioned the collaterals to two types: full and partial collaterals. A finer partition of collateral types distinguishes between two types of partial collaterals: zero and positive collaterals. The next result, Lemma \ref{thm:zero-collaterals-1}, characterizes the cases of zero-collaterals; it defines the condition on a set $A_k$ of investors in enterprise $k$ such that another player $i$ will prefer to cooperate even with a zero collateral.  
Intuitively, the lemma requires that even in the worst case for $i$, in which all the other players that are not in $A_k$ defect, the enterprise return to $i$ when investing will be weakly higher than $x_{ki}$, and so $i$ does not need a collateral to incentivize her to prefer investing. 
The proof of this lemma, as well as other technical proofs in the rest of this paper, are in Appendix \ref{sec:appendix-additional-proofs}.

\begin{lemma}\label{thm:zero-collaterals-1}\textbf{(zero-collateral condition)}:
Denote by $K$ the set of players with investment opportunities in enterprise $k$.
Fix a set $A_k \subset K$  
and any player $i \in K \setminus A_k$. If  
\begin{equation}\label{eq:zero-collateral-condition}
\small
	x_{ki} + \sum_{\substack{j \in A_k}}x_{kj} \geq Z_k(1 + 1/\alpha_k)
\end{equation}
\normalsize
\noindent
then 
for any action profile of all players in which $i$ does not default and the players in $A_k$ invest in $k$, player $i$ has a strict best reply to cooperate w.r.t. $k$ \emph{even with zero collateral ($c_{ki} = 0$)}. 

\vspace{3pt}
\noindent
In addition, if player $i$ with $c_{ki}=0$ has 
utility at least $x_{ki}$ assuming that the players $A_k \cup \{i\}$ invest and the other players in $K$ defect, then 
Equation (\ref{eq:zero-collateral-condition}) must hold. 
\end{lemma}
The next lemma shows a condition on a set $A_k$ of investors in enterprise $k$ such that a player outside the set which is not in default 
prefers to cooperate w.r.t. $k$, 
only if given full collateral. 
Intuitively, in the worst case for this player, where all the other players defect, the enterprise return to the player if she invests is zero, and so a full collateral is needed to incentivize investment. 

\begin{lemma}\label{full-collaterals}\textbf{(full-collateral condition)}:
Denote by $K$ the set of players with investment opportunities in enterprise $k$.
Fix a set $A_k \subset K$  
and any player $i \in K \setminus A_k$. If  
\begin{equation}\label{eq:full-collateral-condition}
\small
	x_{ki} + \sum_{\substack{j \in A_k}}x_{kj} \leq Z_k
\end{equation}
\normalsize
\noindent
then, assuming that the players $A_k$ invest and the other players in $K$ 
defect, player $i$ prefers to cooperate w.r.t. $k$ \emph{only if she has a full collateral ($c_{ki} = x_{ki}$)} and she is not in default.  

\vspace{3pt}
\noindent
In addition, if for every partial collateral $c_{ki} < x_{ki}$ player $i$ has 
utility less than $x_{ki}$ assuming that the players $A_k \cup \{i\}$ invest and the other players in $K$ defect, then 
Equation (\ref{eq:full-collateral-condition}) must hold. 
\end{lemma}

Using Lemma \ref{thm:zero-collaterals-1} and Lemma \ref{full-collaterals} we show that for integer input, 
for any large enough return rates, in any optimal solution all collaterals must be either full or zero. 
This turns out useful in many later proofs. 
\begin{lemma}\label{large-alpha}\textbf{(large $\alpha$)}:
Consider the collateral minimization problem with integer input $\{n, \mathcal{X},$ $\mathcal{Z}, \mathcal{A}\}$ and let $k \in V_e$ be an enterprise vertex.   
If $\alpha_k > Z_k$ then every solution 
 $\textbf{c}$ to the collateral minimization problem satisfies that for each $i$, $c_{ki}\in \{0,x_{ki}\}$. 
Thus, if
 $\alpha_k > Z_k$ for every enterprise $k$, then 
in every optimal solution every positive collateral is a full collateral.
\end{lemma}
%

\section{Single Enterprise}\label{sec:single}
In this section we discuss the case of networks with a single investment enterprise and $n$ potential investors. We call this type of network a \emph{star graph}; i.e., a network with a single center vertex connected with outgoing edges to $n$ vertices, as illustrated in Figure \ref{fig:star-graph}. 
Note that in this single-enterprise setting, for simplicity, we slightly abuse notation and use $n$ to indicate the number of potential investors (not including the single enterprise vertex that does not make an investment decision) and omit the index of the enterprise, keeping only the investor index $i$ of the edge weight $x_i$. In addition, since in the star network the matrices $\mathcal{X}, \textbf{c}$ have only one row with non-zero entries, we consider them as vectors instead of matrices. Additionally, as there is only one enterprise, we use $Z$ to denote its cost and $\alpha$ for its interest rate (instead of using $\mathcal{Z}, \mathcal{A}$).

Notice that each player here has a single incoming edge and so an order of elimination of player strategies is equivalent to an order over the players. Also, in the star graph players cannot be in default since they do not have their own enterprises. 
As the collateral vector that gives all investors full collateral is clearly a viable collateral vector (as there are no defaults), there exists a solution to the single-enterprise collateral minimization problem. We next characterize such solutions.

The next lemma shows an interesting property of the single enterprise setting:
given a set of cooperators, once a single player outside the set
prefers to cooperate with zero collateral, all others will prefer to
cooperate with zero collateral as well. This property will be useful
for our proofs in the this setting. 

\begin{lemma}\label{thm:zero-collaterals-2}\textbf{(one in for free and the rest follow)}:
Fix a set $A$ of players and assume that every player in $A$ cooperates. If there exist a player $i \notin A$ with zero collateral ($c_i = 0$) with a strict best reply to cooperate for any action profile of the other players, $A^c \setminus \{i\}$, then if the players $A \cup \left\{i\right\}$ cooperate, every other player has a strict best reply to cooperate, even with zero collateral.  
\end{lemma}

Proposition \ref{thm:uniqueness} showed that for any collateral matrix that induces a unique Nash equilibrium there is at least one possible order of iterated elimination of dominated strategies in the induced game which results in this unique equilibrium. The following result, 
Proposition \ref{thm:collateral-costs}, works in the inverse direction: for any possible order over the players in a star graph the theorem provides a recipe for constructing a minimal-sum collateral vector such that this order will be an order of iterated elimination of dominated strategies. It does so by giving each player the minimal amount collateral that induces her to cooperate assuming all players before her in the order cooperate, and all the others defect.  

\begin{figure}[t]
	\centering
		\includegraphics[width=0.3\textwidth]{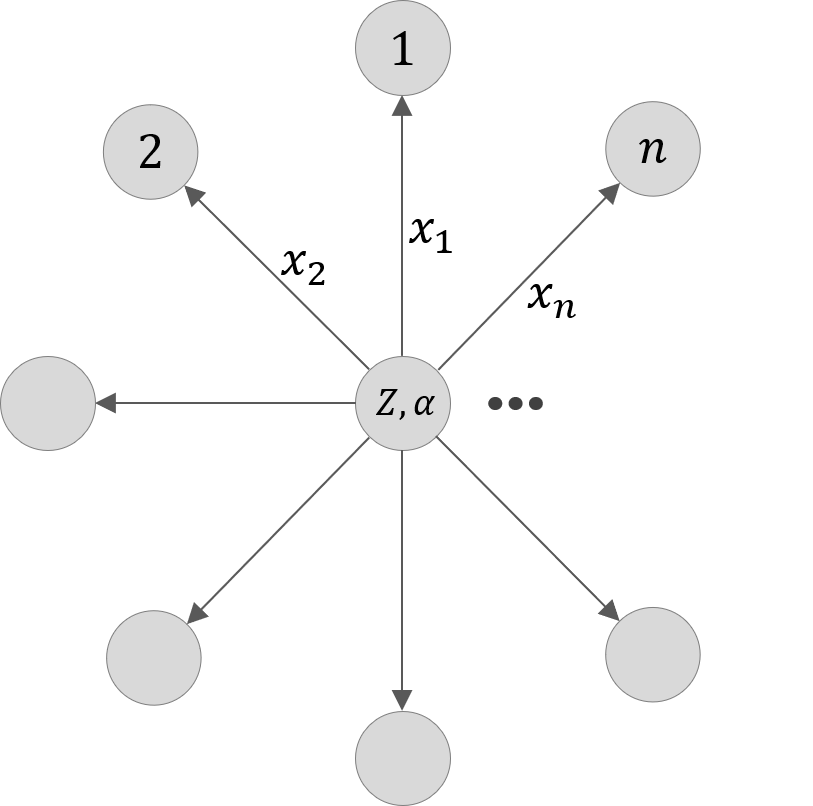}
	\caption{{\small A single-enterprise network.}}
	\label{fig:star-graph}
\end{figure}

\begin{proposition}\label{thm:collateral-costs} \textbf{(minimal collaterals)}:
Let $\Sigma$ be the set of all permutations 
of $[n]$. For every $\sigma \in \Sigma$ there exists a unique minimal collateral vector such that $(\sigma_1, \sigma_2,...,\sigma_n)$ is the order over the players of iterated elimination of strictly dominated strategies, and for every $i\in[n]$ the collateral to player $\sigma_i$ is 
{\small
\begin{equation}
	c_{\sigma_i} = x_{\sigma_i} \cdot \max \Big(0, \min \Big(1, 1-(1+\alpha)\big[1-\frac{Z}{\sum_{\substack{\sigma_j \leq \sigma_i}}x_{\sigma_j}}\big]\Big)\Big). 
\end{equation}
}
\end{proposition}

On the one hand, if all-cooperate is a unique equilibrium, it is reachable by iterated elimination of dominated strategies with some permutation of the players in the star graph, and on the other hand, Proposition \ref{thm:collateral-costs} shows how to find a minimal collateral vector for every permutation.  
Thus, we see that by going over all permutations we reach the entire 
set of minimal collateral vectors. This shows that the set $\mathcal{C}_{min}$ of all minimal collateral vectors, which is the search domain of the collateral minimization problem in the star graph, is a finite set. Specifically, the cardinality of the set $\mathcal{C}_{min}$ is bounded by the cardinality of the set of player permutations. I.e. for any given instance of the collateral minimization problem with $n$ players $|\mathcal{C}_{min}| \leq n!$. 
Later in this section we show that $|\mathcal{C}_{min}|$ is in fact somewhat smaller than $n!$, and bounded by $2^n$.

The following result, Theorem \ref{thm:opt-partial-collaterals}, characterizes the partial collaterals in an optimal solution of the collateral minimization problem with a single enterprise. The theorem shows that the set of players who receive full collaterals in an optimal solution provides sufficient information to easily determine a full solution. 

Another observation arising from Theorem \ref{thm:opt-partial-collaterals} is that if we follow an order of iterated elimination of dominated strategies in an optimal solution, and observe the collaterals given to the players along this order, then initially some subset of the players get full collaterals. These are the players for which \emph{cooperate} is a dominant strategy and so they may have any arbitrary order. 
After these full collaterals, 
the following players, who get partial collaterals,  
are ordered in a decreasing order of their investments. Their partial collaterals are monotone with their investments both in their absolute sizes and in percentages of the investment of each individual; i.e., for partial collaterals the collateral per dollar of investment is decreasing with the investment size (see Corollary \ref{thm:decreasing-partial-collaterals-corollary}). 

\begin{theorem}\label{thm:opt-partial-collaterals} \textbf{(optimal partial collaterals)}:
Assume without loss of generality that players are indexed so that $x_1 \geq x_2 \geq ... \geq x_n$.
Fix any optimal solution of the single-enterprise collateral minimization problem and assume that $A$ is the set of players receiving full collateral in that solution.
Denote $X_A=\sum_{i\in A}x_i$. There exists an optimal solution in which players in $A$ have full collateral and 
the collaterals for the other players $B = [n] \setminus A$ are  
{\small
\begin{equation}\label{eq:opt-partial-collaterals}
	\forall i \in B, \ c_i = \max\Big(0, x_i\Big[1 - (1+\alpha)\Big(1-\frac{Z}{X_A + \sum_{\substack{j \leq i\\j \in B}}x_j}\Big)\Big]\Big).
\end{equation}
}
\end{theorem}

The proof is in Appendix \ref{sec:appendix-additional-proofs}. The proof idea is that Equation (\ref{eq:opt-partial-collaterals}) is obtained from Proposition \ref{thm:collateral-costs} under the condition that in the order $\sigma$ of iterated elimination of dominated strategies, players in $A$ appear first (in an arbitrary order) and are followed by the players in $B$, ordered in an increasing order of their indices. The proof then shows that whenever this condition does not hold we can always reduce the sum of all collaterals, in a weak sense, by ordering any pair of consecutive players in the order of iterated elimination who have partial collaterals such that the player with the smallest index (and so a larger or equal investment amount) comes first.

An easy corollary of the theorem is that partial collaterals are monotone in investments, both in their absolute sizes and in percentages of the investment of each individual.
\begin{corollary}\label{thm:decreasing-partial-collaterals-corollary}
For two players $a,b$ that have partial collaterals in an optimal solution of the single-enterprise collateral minimization problem, if $x_a > x_b$ then $c_a > c_b$ and $c_a/x_a > c_b/x_b$. 
\end{corollary}

Theorem \ref{thm:opt-partial-collaterals} shows that the collateral minimization problem can be reduced to the problem of selecting a subset of players for which to give full collaterals: Because there exists an  optimal solution in which some set $A$ of players have full collaterals and the other collaterals are given by Equation (\ref{eq:opt-partial-collaterals}), one needs only to search for the set $A$.
This is formalized in the following corollary. 
\begin{corollary}\label{thm:star-solution-procedure}
A solution to the single-enterprise collateral minimization problem with $d$ players can be found in $\mathcal{O}(2^d d)$ time by the following procedure: 
For every subset $S$ of the players construct an ordering $\sigma(S)$ in which the players in $S$ are the prefix (in an arbitrary order) and the other players after the prefix, $S^c$, are ordered in 
a non-increasing order of investment amounts; 
then, set the collaterals for the ordering $\sigma(S)$ with full collaterals to $S$ and the collaterals to $S^c$ according to 
Theorem \ref{thm:opt-partial-collaterals}, resulting in total collateral of $c(S)$.
Finally, select a solution obtained for a set $S$ with minimum total $c(S)$. 
\end{corollary}

\subsection{Computational Complexity}\label{subsec:complexity}
In the previous section we have seen that the single-enterprise collateral minimization problem can be solved as a subset-selection problem for the set of players that receive a full collateral for their investment. 
In this section we present a formal connection between the single-enterprise collateral minimization problem and the knapsack problem, proving that the single-enterprise collateral minimization problem is NP-hard for integer inputs. 
The outline of the proof is as follows. 
First we consider the regime of a large interest parameter $\alpha$ in which, by Lemma \ref{large-alpha}, all the collaterals are either full or zero. Then we show that in such instances the largest investor (or, in case there are repetitions of the same largest values, one of those investors) must receive a zero collateral in an optimal solution (Lemma \ref{largest-player}). We then define an inverse version of the uniform knapsack problem in which we seek a set of items of one-dimensional sizes that together exceed a given capacity threshold, but do so in the minimum way (Definition \ref{inverse-knapsack}). We show that this ``inverse knapsack problem'' is equivalent to the standard uniform knapsack problem. Then we show a reduction from the inverse knapsack problem to our problem.
\begin{lemma}\label{largest-player}\textbf{(largest player)}:
Consider the single-enterprise collateral minimization problem with integer  input $\left\{n, \left\{x_i\right\}_{i=1}^n, Z, \alpha \right\}$. If $\alpha > Z$
 then in every optimal solution at least one of the players with the largest investment opportunity receives a zero collateral. 
\end{lemma}

Next, we define the \emph{inverse knapsack problem} and observe that it is NP-hard, and then 
reduce it to the collateral minimization problem for star graphs with integer input.
\begin{definition}\label{inverse-knapsack}\textbf{(inverse knapsack)}:  
	For an input of integers $\left\{x_i\right\}_{i=1}^n$ and an integer threshold $t$, the \emph{inverse knapsack problem} asks to find 
	{\small
		$S^* \in \underset{{\tiny S \subseteq [n]}}{\argmin}  \sum_{i \in S}x_i, \ s.t. \
		\sum_{i \in S}x_i > t$.
	}
\end{definition}   
\begin{lemma}\label{knapsak-iff-inverse-knapsack-new}	
	The inverse knapsack problem is NP-hard.
\end{lemma}
\begin{theorem}\textbf{(hardness)}:
The single-enterprise collateral minimization problem with integer input $\left\{n, \left\{x_i\right\}_{i=1}^n, Z, \alpha \right\}$ is NP-hard.
\end{theorem}
\begin{proof}
The proof is by reduction of the inverse knapsack problem to the  single-enterprise collateral minimization problem.
Let $\left\{x_i\right\}_{i=1}^n$, $t \in \mathbb{N}$ be the input to an inverse knapsack problem. We assume that the inputs are all integers and also that $t \leq \sum_i x_i - \max_i{x_i}$ (otherwise, the largest player must be inside the solution set and the problem can be mapped to a new problem without this player and with a lower threshold). 
Next we construct a special case of our problem which solves the inverse knapsack problem.

Consider an instance of the collateral minimization problem with the set $[n]$ of players with investments $\left\{x_i\right\}_{i=1}^n$ and an additional player, $n+1$, with investment $x_{max} = \max_i{x_i} + 1$, and $\alpha = 2Z$. 
By Lemma \ref{large-alpha} since $\alpha > Z$, collaterals in an optimal solution are either full or zero. 
By Lemma \ref{thm:zero-collaterals-2} we know that a single player that strictly prefers to cooperate with zero collateral assures that the game induced by the collaterals has a unique Nash equilibrium, and by Lemma \ref{largest-player}, in an optimal solution this player is the largest player. 
In an optimal solution, the set $S$ of players with full collaterals, is a set with minimum sum that satisfies the constraint 
{\small $
\sum_{j \in S} x  > Z - x_{max} 
$}.
By setting the parameter $Z$ of the collateral minimization problem to $Z = x_{max} + t$, we obtain:
{\small $
\sum_{j \in S} x_j  > t
$}. 
Because $S$ is a minimum-sum set that satisfies this condition, and since $x_{max}$ that we added to the inverse knapsack input cannot be in $S$, this is a solution for the inverse knapsack problem. Hence, the single-enterprise collateral minimization problem with integer input  is NP-hard.
\end{proof}

\section{Investment Networks}\label{sec:investment-networks}
\vspace{-3pt}
In the previous section we studied investment networks that have very simple structure: star graphs with a single enterprise and $n$ investors. Next we return to the general model of arbitrary directed investment networks. In this setting, as detailed in Section \ref{sec:model-all}, if the enterprise of an agent has defaulted (not enough capital was invested to cover the operational cost), 
the agent withdraws all her investments with zero recovery, and essentially does not make an investment in any enterprise. 
The key difference between the single-enterprise setting and the setting of a network of investments with multiple enterprises is the risk of default cascades. For example, if firm $A$ is a potential investor in the enterprise of firm $B$, and $A$ defaults, then $A$ does not invest in enterprise $B$, regardless of the incentives $A$ had for this investment. This in turn may lead firm $B$ to a default as well, in which case $B$ cannot invest in any enterprise.

In this section we present several results for the network model. 
We first show that in a \emph{directed acyclic graph (DAG)} an optimal collateral solution that is sufficient to induce a unique Nash equilibrium in each enterprise, ignoring the risk of default cascades, is actually sufficient to induce a unique Nash equilibrium in the entire network, even though defaults trigger investment failures. In fact, these locally optimal solutions that are calculated for each enterprise separately yield a globally optimal solution in the network. 
We thus conclude that for DAGs, no \emph{Network Excess Collaterals (NEC)} are needed, that is, no extra collaterals are needed due to systemic risk resulting from the potential for default cascades. As a corollary of this result, we get that if it is possible to solve the collateral minimization problem for each enterprise, we can also solve the problem for a DAG constructed from those enterprises. In particular, in any DAG in which every enterprise has at most $d$ potential investors, we can solve the problem in time  $\mathcal{O}(2^d dn)$.

We then move to consider general directed graphs (potentially with cycles) 
and present two negative results. First, we show that default cascades can no longer be ignored, 
and they might imply a huge increase in the total collateral needed to induce  a unique equilibrium (even when exists) -- the increase can be arbitrary large (NEC is unbounded), and this holds even for networks with constant number of agents. Finally, we show that cycles also imply computational hardness, even when every single enterprise problem in isolation can be solved. Specifically, we show that the problem is NP-hard on general graphs, even with access to an oracle that can solve any single-enterprise problem.   

\subsection{Preliminaries}
We first define the concept of a \emph{star decomposition} of an investment network. This concept is useful for comparing local versus network effects in the collateral minimization problem.

\begin{definition} \textbf{(star decomposition)}:
Let $G = (V,E)$ be a directed graph and denote by $V_e = \left\{v \in V: \ \exists u \in V \ s.t. \ (v,u) \in E \right\}$ the subset of vertices with non-zero out-degree in $G$. 
The \emph{star decomposition of $G$}, denoted by $\mathcal{H}(G)$,  is defined as the family of sub-graphs of $G$ such that each sub-graph includes a vertex $c \in V_e$ and all its neighboring vertices $\left\{v \in V : \ (c,v) \in E\right\}$ and the edges $\{(c,v) \in E\}$ connecting $c$ to them. 
\end{definition}
Notice that in our definition of the star decomposition, the same vertex may appear in multiple star sub-graphs if it considers multiple investments. See Figures \ref{fig:DAG-example} and \ref{fig:cycle-example} for examples of graphs and their star decompositions.
The star decomposition naturally induces a set of single-enterprise collateral minimization problems, one for each sub-graph (with the corresponding parameters).

To study the effect of the topology of the investment network on the optimal collateral scheme, we define the 
\emph{Network Excess Collaterals (NEC)} 
of a network for which a viable collateral matrix exists, to be the ratio between the sum of the optimal collaterals in the network and the sum of optimal collaterals in its star decomposition.   
Clearly, the cost of stabilizing each star sub-graph when it is a part of a network can only be weakly higher than the cost if stabilizing the same star graph when it is not connected to a network. The NEC quantifies the \emph{systemic} strategic risk component of the cost of stabilizing the network in a unique equilibrium, or the excess collaterals needed to overcome the risk of default cascades.     

For a collateral minimization problem $P=\{n, \mathcal{X}, \mathcal{Z}, \mathcal{A}\}$ with investment graph  $G = (V,E)$ we will use $P_G$ to denote the problem, to make it explicit that it has an investment graph $G$. If $H_i$ is a subgraph of $G$, we use $P_{H_i}$ to denote the collateral minimization problem induced by $P_G$ on the subgraph $H_i$. For any problem $P_K$ with investment graph $K$, we denote by $C_K$ the total collateral in any optimal solution of $P_K$ when such exists, and $C_K=\infty$ otherwise.

\begin{definition}
	Consider the collateral minimization problem $P_G=\{n, \mathcal{X}, \mathcal{Z}, \mathcal{A}\}$ with graph  $G = (V,E)$.
	Let $\mathcal{H}(G)$ be the star decomposition of $G$. 
	The Network Excess Collaterals (NEC) of $P_G$ is defined as:
	\begin{equation}\label{eq:nec}
	NEC(P_G) = \frac{{C_G}}{\sum_{H_i \in \mathcal{H}(G)} {C_{H_i}}}
	\end{equation}
\end{definition}

Note that as any $H_i\in \mathcal{H}(G)$ is a star graph, $C_{H_i}$ is always finite and hence the denominator is always finite.

\subsection{Directed Acyclic Graphs}
\begin{figure}[t]
	\centering
		\includegraphics[width=0.6\textwidth]{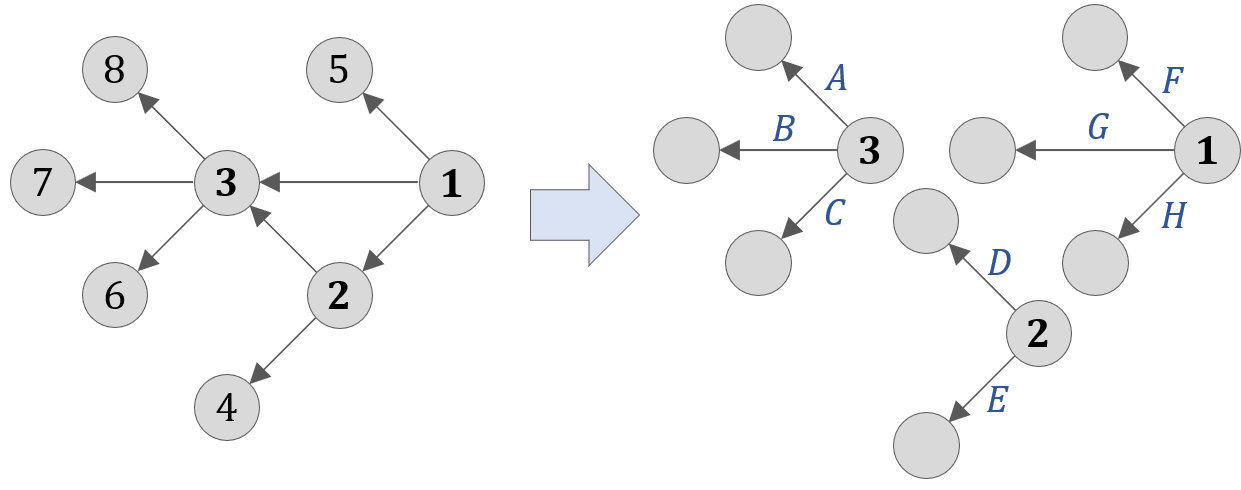}
	\caption{
	{\small
	Left: A directed acyclic graph (DAG). The numbers show a topological order. Vertices 1, 2, 3 are the enterprise vertices. Right: The star decomposition of the graph. The topological order induces an order over the star sub-graphs. In an iterated elimination of dominated strategies of the players (each decision corresponds to an edge), edges of star sub-graphs that are higher in the order are eliminated first. The capital letters show an example for such order of iterated elimination. 
The iteration starts from the edges of star graph number 3 in the topological order, which are ordered by a solution of this star graph problem. Once star graph 3 is secured in a unique equilibrium, star graph 2 can be solved as an independent problem, without risk of default from star graph 3, and the process continues down the order of the star graphs that is induced by the topological order. 
}}
\label{fig:DAG-example}
\end{figure}

We start by considering investment graphs that are acyclic. 
Our main result for this setting is that the optimal collaterals for an acyclic network of investments can be obtained by solving each enterprise in the star decomposition of the investment graph separately, ignoring the fact that it is a part of a network. 
These collaterals are clearly enough to induce a unique equilibrium when ignoring the possibility of default cascades. 
The main insight of the proof is that we can actually find an order of elimination of dominated strategies that will ensure that when an agent considers investing she is reassured that no default can harm her, even through a cascade. Essentially, in every DAG there is an enterprise vertex that has no investors which are enterprise vertices themselves. So we can argue we can put this enterprise vertex and its investors first in the order of elimination of dominated strategies (with order and collaterals according to the optimal solution for this enterprise as a stand-alone enterprise).
Any later agent knows that this first enterprise will not default, so we can then continue recursively on the rest of the network, without fear that the 
previous enterprises will trigger a default cascade. See Figure \ref{fig:DAG-example} for an illustration of the process. 

\begin{theorem} \textbf{(optimal collaterals in DAGs)}: \label{optimal-collaterals-in-DAGs}
Consider the collateral minimization problem $P_G=\{n, \mathcal{X}, \mathcal{Z}, \mathcal{A}\}$ with graph  $G = (V,E)$.
Let $\mathcal{H}(G)= \{H_i\}_{i\in V_e}$ be the star decomposition of $G$, where $V_e$ is the set of enterprise vertices. 
If $G$ is acyclic then a solution to the collateral minimization problem $P_G$ can be obtained by 
solving the collateral minimization problem induced for each of the sub-graphs in $\mathcal{H}(G)$ separately. 
Thus, for any problem $P_G$ with graph $G$ that is a DAG, it holds that $C_G=\sum_{H_i \in \mathcal{H}(G)} C_{H_i}$ and so $NEC(P_G) = 1$.
\end{theorem}
\begin{proof} 
	For each sub-graph $H_i \in \mathcal{H}(G)$, 	 
	compute a solution to the collateral minimization problem induced for this single enterprise graph, and let $C_{H_i}$ be the total collateral in that solution. 
	We claim that these same collaterals (with an appropriate order of elimination to be specified below) 
	form a solution to the problem $P_G$. This solution has total collateral  $C_G=\sum_{H_i \in \mathcal{H}(G)} C_{H_i}$, and we argue that no solution has lower total.
	First, if there was a solution for $P_G$ with lower total, it would induce a solution for one of the sub-graphs $H_i$ (with the same elimination order and collaterals), with lower total than $C_{H_i}$, which is a contradiction. 
	Next, we argue that there is indeed a solution for $P_G$ with total $\sum_{H_i \in \mathcal{H}(G)} C_{H_i}$. To prove the claim we present an order of elimination of dominated strategies for which these collaterals work.  
	
Let $\xi_G$ be a topological ordering of $G$, an order of the vertices in which for every directed edge $(i,j)$, vertex $i$ appears before $j$. Such an ordering exists as $G$ is a DAG. 
Consider the following order over player strategies -- enterprise vertices (vertices in $V_e$) are ordered in reversed order of the topological order; when considering an enterprise vertex $i\in V_e$ in that order, we process its investment opportunities (edges) 
in the order defined by the solution to the star graph $H_i$ of this enterprise vertex. This defines an order on decision edges for the elimination of dominated strategies for $P_G$. We argue that if we eliminate dominated strategies by that order, all players will indeed invest in all their investment opportunities. Indeed, it is easy to see by induction that when a vertex considers investing in an enterprise, the above order ensures that all other investors in the same enterprise will never default under the iterative elimination of dominated strategies. See Figure \ref{fig:DAG-example} for an example of the process. 
\end{proof}

Theorem \ref{optimal-collaterals-in-DAGs} shows that every collateral minimization problem with acyclic investment structure can be broken down into $\mathcal{O}(n)$ star-graph problems, and solving those provides a solution to the network problem. 
If the network has a bounded out-degree $d$, the procedure shown in the previous section for solving the collateral minimization problem in star graphs has $\mathcal{O}(2^d d)$ running time (Corollary \ref{thm:star-solution-procedure}). Hence we have the following corollary. 
\begin{corollary}\label{bounded-degree-DAGs-solution-time}
Consider the collateral minimization problem $P_G= \{n, \mathcal{X}, \mathcal{Z}, \mathcal{A}\}$ with graph  $G = (V,E)$.
	If $G$ is a DAG and has a bounded out-degree $d$, then the collateral minimization problem $P_G$ can be solved in $\mathcal{O}(2^d dn)$ time. In particular, for constant out-degree graphs the problem can be solved in polynomial time. 
\end{corollary}

\subsection{General Graphs}\label{sub-sec:general-graphs}
Theorem \ref{optimal-collaterals-in-DAGs} showed that in every directed \emph{acyclic} graph $G$ it holds that for every collateral minimization problem $P_G$ with the topology $G$, $NEC(P_G)=1$, that is, no extra collaterals are needed to handle the network systemic risk. 
In sharp contrast, the following theorem shows that in graphs with directed \emph{cycles}, the $NEC$ is unbounded, 
even when only considering instances for which it is possible to induce a unique equilibrium by collaterals.  

The proof uses the construction
depicted in Figure \ref{fig:cycle-example}. 
The figure shows a cycle of investment opportunities in which every cycle vertex also has two external investors in
its own enterprise, one of which is with a fixed investment amount (of size $1$), and the other with an investment amount parametrized
by $k$. The proof shows that to establish full investment as a unique equilibrium, full collaterals must be given to both external investors
of at least one of the cycle vertices, leading to a total collateral greater than $k+1$, which can be made arbitrarily large (by increasing $k$). The basic
intuition is that these cycle vertices can never have a dominant strategy (even with full collaterals, since they may still default), and so an order of iterated elimination of dominated strategies can be established only if one of them is secured from default.

\begin{theorem} \textbf{(cycle costs)}:
For any $R>1$ there exists a solvable collateral minimization problem $P_G=\{n, \mathcal{X}, \mathcal{Z}, \mathcal{A}\}$ with cyclic investment graph $G = (V,E)$ for which $NEC(P_G)>R$. 
Moreover, this holds even for problems that all share the same investment graph $G$ with only one cycle with only $3$ enterprise vertices 
 and a constant number of vertices 
($n=|V|= 9$). 
\end{theorem}

\begin{proof}
We present a cyclic network proving the claim, see Figure \ref{fig:cycle-example} for an illustration. On the left, the figure illustrates a symmetric cycle network with nine players, three of which are investors that also have investment enterprises, and the other six are investors with no enterprises of their own. The potential investment amounts are depicted on the edges of the graph, where the arrows 
are in the direction of the returns of investments. 
We will use the investments of size $k > 2$ as a parameter to be picked later, and set 
$Z_A = Z_B = Z_C = k+1$ 
and $\alpha = 2k$.

In the optimal solution of a star graph in the star decomposition of $G$, say for sub-graph $A$, the investors $B$ and $a_1$ receive full collaterals and $a_2$ cooperates with a zero collateral. This gives in the star decomposition a cost of $2$ for each enterprise, and a total cost of $6$. 
In the cycle, a viable collateral matrix exists (e.g., full collaterals to all players form a viable collateral matrix), but these collaterals of the star decomposition are not sufficient: in any order of iterated elimination of dominated strategies none of the investment opportunities of players $a_2, b_2, c_2$ can be before the other two, unless this investment is secured with a full collateral. Thus, in the cycle the total cost is $k+1$ for one of the star sub-graphs and $2$ in the other star sub-graphs, with a total of $k+5$.
Thus, the ratio of costs between the optimal solution of the cycle graph and the sum of locally optimal solutions of the star decomposition is $\frac{k+5}{6}$. So we can pick $k$ to be large enough so that $\frac{k+5}{6}>R$ to prove the claim.  
This network demonstrates that in graphs with cycles, the global optimum can be arbitrarily more expensive than the sum of local solutions, with a NEC of $\Theta(k)$.
\end{proof}

\begin{figure}[t]
	\centering
		\includegraphics[width=0.94\textwidth]{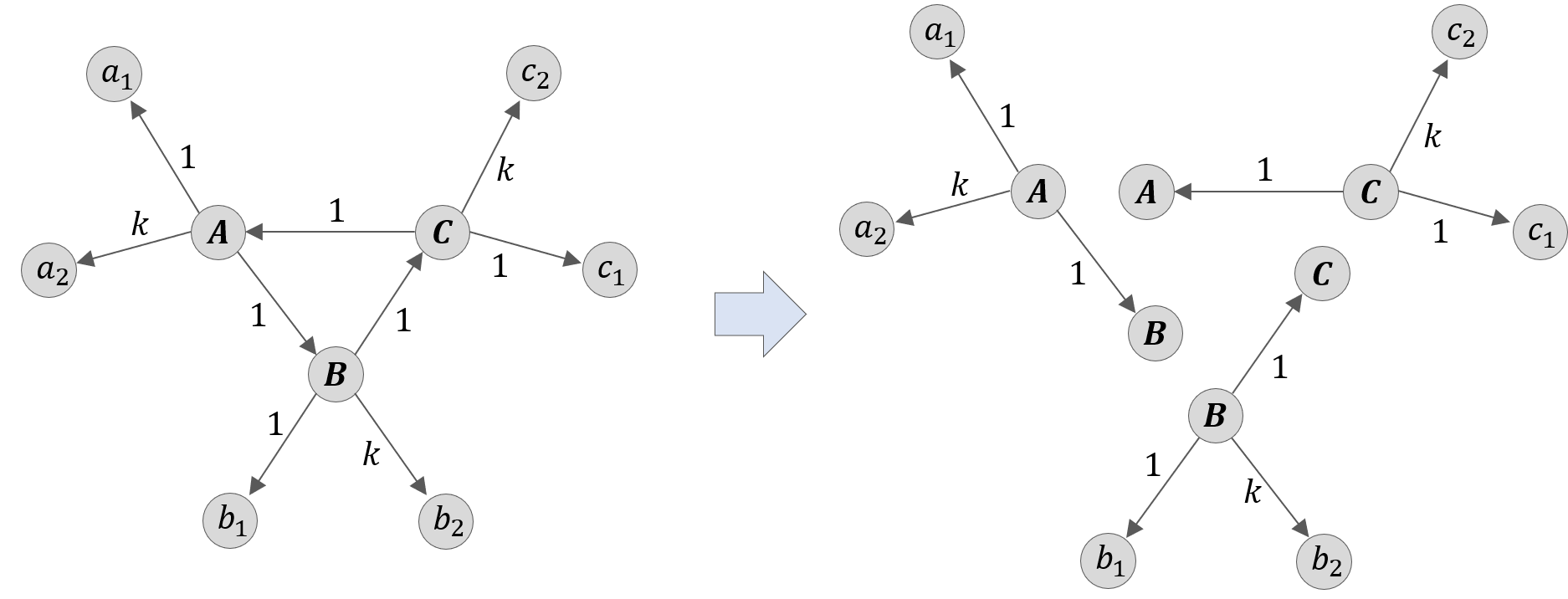}
	\caption{{\small Left: A cycle graph; Right: the star decomposition.}} 
	\label{fig:cycle-example}
\end{figure}

The following theorem shows that for a family of instances for which a solution always exists,\footnote{Cases in which a cyclic investment network has no viable collateral matrices can be identified in polynomial time. See Section \ref{sec:unique-eq-through-collaterals} and Appendix \ref{appendix:existence-of-viable-collateral-matrices}. 
}
cycles in the investment network imply computational hardness, 
even when local solutions to each single enterprise are known.   

\begin{theorem} \textbf{(hardness due to cycles)}:
The collateral minimization problem with integer input $\{n, \mathcal{X}, \mathcal{Z}, \mathcal{A}\}$ for cyclic investment networks is NP-hard. Moreover, it is so even given an oracle to solve the problem of optimal collaterals on star graphs.
\end{theorem}

\begin{proof}
The proof is by reduction from the feedback vertex set problem. Let $G=(V,E)$ be a directed graph.
Denote the out-degree of vertex $i\in V$ by $d_i$, and define $k = 1 + \max_{i\in V} d_i$.
Now we wish to construct an instance of the collateral minimization problem for which the solution solves also the optimization (minimization) version of the feedback vertex set problem.
For this we will construct a graph $G' = (V',E')$ by the following steps:
First, we iteratively remove from $G$ all the vertices that have a zero out-degree (and incident edges) until there are no such vertices. This process takes polynomially many steps, and results in a graph that contains only the vertices ${V_0}\subseteq V$ that are cycle-vertices in $G$, and all edges induced by ${V_0}$. Notice that the cycles in the resulting graph are exactly the same as in the graph $G$ (and so the solutions to the feedback vertex set problem on both graphs are the same). Then, we assign each edge $(i,j)$ for $i,j\in V_0$ a weight of $x_{ij} = 1$. 
Next, we add to each vertex  in $V_0$ a ``gadget'' with two neighbors -- new vertices with a zero out-degree and in-degree of one each. One vertex has an incoming edge with weight $1$, and the other  with weight $k$. 
Let $G' = (V',E')$ be the resulting graph ($V'$ is the union of $V_0$ with the additional $2|V_0|$ vertices from the gadgets).
Now we assign each vertex $i \in V_0$  
 two parameters: 
$Z_i= k + 1$ and $\alpha_i = 2k$. 
Notice that all the cycles in the graphs $G$, and $G'$ are the same cycles since we only added ``spike'' vertices that do not affect cycles, and thus the minimum feedback vertex set of $G'$ is the same as in $G$.
The weighted graph $G'$ together with the vertex properties $\{Z_i\}$, $\{\alpha_i\}$ also define an instance, $P_{G'}$, of the collateral minimization problem. In this instance of the problem, viable collateral matrices exist (e.g., full collaterals to all players are a viable collateral matrix) and optimal collaterals have the following properties:

\vspace{5pt}
\noindent
\emph{Claim:} In optimal collaterals of $P_{G'}$, in each cycle there is at least one vertex for which the sum of collaterals for its outgoing edges is $k+1$. In addition, the sum of collaterals to outgoing edges of every vertex with a non-zero out-degree is either $2$ or $k+1$. 

\vspace{5pt}
\noindent
\emph{Proof:} As $\alpha_i=2k>k+1=Z_i$ for every $i\in V_0$, by Lemma \ref{large-alpha} this instance is in the large $\alpha$ regime, implying that every collateral is either  full or zero. 
Thus, for each edge $(i,j)$ for nodes $i,j\in V'$, the collateral $c_{ij}$ equals either $x_{ij}$ or zero. 
Consider an arbitrary cycle in $G'$. 
For each vertex in the cycle and its outgoing edges (which define a star sub-graph) if there is no risk of default then the optimal collaterals are full collaterals to two players with weight 1 (as now the weight $k$ player in the star sub-graph prefers to invest, even with zero collateral, and then all others follow). 
However, if these are the collaterals in a cycle, they are not sufficient to get a unique equilibrium in iterated elimination of dominated strategies. This is so
because no player of weight $k$ can come next after these cooperating size-$1$ players in any order of iterated elimination, unless she gets a full collateral, due to a risk of default somewhere in the cycle (and the same holds true for any others without collaterals). On the other hand, once one of these size-$k$ players in the cycle has a full collateral, another full collateral must be given to a player of size $1$ in the same star sub-graph to cover the operational costs, and then no further collaterals are needed in this star sub-graph. 
Thus, in an optimal solution, in each cycle  there is a vertex for which the sum of collaterals to its outgoing edges is $k+1$.
Now assume a star sub-graph in which in the optimal solution the sum of collaterals is less than $k+1$. A necessary and sufficient condition that these collaterals induce sufficient cooperation to cover the operational cost of $Z=k+1$ is that the player of size $k$ can be convinced that at least two players of size $1$ are in the set of investors, in which case the size-$k$ player has profits from cooperating with zero collateral. Because the solution is optimal, the sum of collaterals in such star sub-graph cannot be greater than $2$.

\vspace{3pt}
Denote by $S^*$ a set of vertices 
such that for each such vertex the sum of collaterals for its outgoing edges is $k+1$.
In a solution of the collateral minimization problem each cycle contains at least one vertex from $S^*$, and so the sum of collaterals can be written as $(k+1)|S^*| + 2|(S^*)^c| $. Because the solution is optimal, we have 
{\small
\begin{equation*}
 S^* \in  \underset{{\tiny S \subseteq V'}}{\argmin} |S| \quad s.t. \ every \ cycle \ in \ G' \ contains \ a \ vertex \ in \ S^*.
\end{equation*}
}
Because $G$ has the same cycles as $G'$,  
{\small
\begin{equation*}
S^* \in  \underset{{\tiny S \subseteq V}}{\argmin} |S|  \quad s.t. \ every \ cycle \ in \ G \ contains \ a \ vertex \ in \ S^*, 
\end{equation*}
}
which is a solution of the feedback vertex set problem.
Hence, the collateral minimization problem is NP-hard even if the solutions for star-graph problems are known. 
\end{proof}

\vspace{3pt}
The above results shed light on the potential effect of cycles in investment networks and their effect on systemic risk:  
First, networks without cycles are easier to stabilize with collaterals; the existence of cycles generally increases the minimum cost of collaterals and adds computational difficulty even in bounded-degree networks. Second, in networks without cycles, the optimal collaterals that each firm optimizes locally to offer to its potential investors also yield the global optimum of collaterals that would be calculated by a central planner. In contrast, when cycles exist there are networks in which the locally-optimal collaterals do not lead to a unique equilibrium, and so central planning and interventions may be required to ensure stability.

\subsection*{Acknowledgments}
This project has received funding from the European Research Council (ERC) under the European Union's Horizon 2020 Research and Innovation Programme (grant agreement No. 740282).

\bibliography{Optimal-Collateral-bib}
\appendix
\section*{APPENDICES}
\section{Viable Collateral Matrices}\label{appendix:existence-of-viable-collateral-matrices}
\begin{lemma}\label{thm:viable-collateral-matrices-existence}\textbf{(unique equilibrium with full collaterals)}: 
Consider the input $\{n, \mathcal{X}, \mathcal{Z}, \mathcal{A}\}$ with graph $G = (V,E)$ to the collateral minimization problem.
Viable collateral matrices exist for this input (and thus, a solution to the collateral minimization problem exists) if and only if 
the all-cooperating strategy profile is reached in  
the following polynomial-time procedure: 
(i) Set the collateral for every investment opportunity to a full collateral, i.e., $\textbf{c} = \mathcal{X}$; (ii) Iteratively eliminate strictly dominated strategies. 
\end{lemma} 

\begin{proof}
If the iterated elimination with the full collaterals, $\textbf{c} = \mathcal{X}$, reaches the all-cooperating strategy profile, then full collaterals form a viable collateral matrix, and hence a solution to the collateral minimization problem must exist. 
If the iterated elimination does not reach the all-cooperating strategy profile, then, under full collaterals to all players, there are some players that prefer defecting over cooperating. 
Assume by contradiction that there is a viable collateral matrix $\hat{\textbf{c}} \ne \textbf{c}$. 
Collaterals that are higher than full collaterals do not have any effect on the utility of any player, and so we can assume w.l.o.g. that $\forall k,i \ \hat{\textbf{c}}_{ki} \leq x_{ki}$, and there exist some investment opportunities $x_{ki}$ such that $\hat{\textbf{c}}_{ki} < \textbf{c}_{ki}$. 
Thus, if we increase the collaterals in $\hat{\textbf{c}}$ to full collaterals there must be players that switch their decision from cooperating to defecting. 
Combined with the monotonicity property, as stated in Proposition \ref{thm:monotonicity}, this leads to a contradiction since the utility of each player from cooperating, as given by Equation (\ref{eq:cooperate-utility}), is a weakly increasing function of her collateral, while her utility from defecting is a constant. 
Notice in addition that a change in collateral for some investment opportunity $x_{ki}$ cannot affect the utility of any player from a different investment opportunity, unless this change in collateral leads player $i$ to change her decision for the investment $x_{ki}$ (the utility from any other investment is a function of the investment profile and not an explicit function of the collaterals and the investment profile is, in turn, a function of the strategy profile and not an explicit function of the collaterals). Thus we conclude that if full collaterals are not a viable collateral matrix, no viable collateral matrix exists.
\end{proof}

\begin{lemma}\textbf{(remaining edges form cycles)}:\label{thm:remaining-vertices-form-cycles}
Consider the process described in Lemma \ref{thm:viable-collateral-matrices-existence} with the input $\{n, \mathcal{X}, \mathcal{Z}, \mathcal{A}\}$ and graph $G = (V,E)$, and denote by $G'$ the sub-graph induced by the set of edges which were not eliminated in that process. 
If $G'$ is not the empty graph, then every vertex in $G'$ is a part of a directed cycle in $G'$.
\end{lemma}
\begin{proof}
Assume by contradiction that there is a vertex $v'$ in $G'$ that is not a part of any directed cycle in $G'$. Then, either $v'$ is a ''dead-end vertex'' in $G'$, i.e., has zero out-degree, or there is a path in $G'$ from $v'$ to a dead-end vertex in $G'$. In either case, the strategies of the dead-end vertex can be eliminated since all its outgoing edges in the original network $G$ are invest edges, and hence she is not in default and prefers to cooperate with a full collateral in all her investments. This is a contradictions since these strategies were not eliminated in the iterated elimination process. Therefore, every remaining vertex after an iterated elimination process with full collaterals is a cycle vertex in the remaining sub-graph $G'$.     
\end{proof}

\begin{lemma}\textbf{(non-satisfiable sub-graphs)}:\label{thm:non-satisfiable-sub-graphs}
It is not possible to induce a unique equilibrium by collaterals if and only if there exists a non trivial sub-graph in which every vertex is a cycle vertex and for every enterprise in that sub-graph the cost of that enterprise is higher than the total investment by all players outside of that sub-graph investing in that enterprise.      
\end{lemma}
\begin{proof}
Assume that a unique equilibrium cannot be induced by collaterals. By Lemma \ref{thm:remaining-vertices-form-cycles}, after an iterated elimination of dominated strategies with full collaterals there remains a sub-graph $G'$ in which every vertex is a part of a cycle in $G'$ and no further strategies (corresponding to the edges of $G'$) can be eliminated. 
Assume by contradiction that there is a vertex $v'$ in $G'$ for which the cost of her enterprise is lower or equal to the total investment by all players outside of the sub-graph $G'$ who consider investments in her enterprise. Due to the iterated elimination of all investments (corresponding to edges) outside of $G'$, specifically, all the outgoing edges of $v'$ to vertices outside of $G'$ are \emph{invest edges}. Thus, the cost of the enterprise of $v'$ is already covered and $v'$ is not in default. Because $v'$ is a cycle vertex in $G'$, she considers an investment opportunity in another enterprise $k'$ in $G'$. As $v'$ is not in default and has a full collateral for this investment opportunity, her strategy for this investment can be eliminated (she prefers to cooperate), in contradiction to $G'$ being the remaining sub-graph after the elimination. 
If, on the other hand, we assume that a unique equilibrium can be induced by collaterals, then
assume by contradiction that there is a non trivial sub-graph $G'$ in which every vertex is a cycle vertex in $G'$ and for every enterprise in $G'$ the cost of that enterprise is higher than the total investment by all players outside of $G'$ investing in that enterprise. Because for every vertex in $G'$ investments from outside of $G'$ cannot cover its operational costs, each vertex in $G'$ is in default if all the other vertices in $G'$ defect w.r.t. all their investment opportunities, and thus prefers defecting for such strategy profiles. Therefore, among all investment opportunities corresponding to edges in $G'$, no investment opportunity can be eliminated first. Due to Proposition \ref{thm:uniqueness}, this is a contradiction to the existence of a unique equilibrium. Thus, if such a sub-graph exists, a unique equilibrium cannot be reached.      
\end{proof}

\noindent
The above results are used for proving Proposition \ref{thm:unique-equilibrium-existence} as follows.\\

\begin{proof} \emph{(Proposition \ref{thm:unique-equilibrium-existence}):} 
We show that the iterative process described in the proposition is equivalent in the investment game to an iterated elimination of strictly dominated strategies with full collaterals, and reaching the empty graph is equivalent to reaching the all-cooperating strategy profile, and thus, by Lemma \ref{thm:viable-collateral-matrices-existence}, viable collateral matrices exist if and only if the iterative process reaches the empty graph. Assume full collaterals to all investments. Non-enterprise players (spikes in the graph) have a dominant strategy to invest in all their investment opportunities. After these strategies are eliminated, 
if we assume it is possible to induce a unique equilibrium by collaterals, then there exists an enterprise-player $k$ that one of her investments is next in an order of iterated elimination, and who already raised sufficient capital to avoid default (or otherwise she would prefer defecting). Due to full collaterals, player $k$ that is not in default prefers investing in all her investment opportunities. 
In this case, decisions that come next in the order of elimination are not affected if we remove $k$ and replace every investment opportunity that $k$ has by the same opportunity from a new spike node, since in both cases these same investments are all successfully performed (without default). The same holds for the next investments in the order of elimination, and once we reach the all-cooperating strategy profile at the end of the iterated elimination, the iterative process described in the proposition has removed all the vertices of the original graph.  
If, on the other hand, we assume that it is not possible to induce a unique equilibrium by collaterals, then, by Lemma \ref{thm:non-satisfiable-sub-graphs} there exists a non trivial sub-graph in which every vertex is a cycle vertex and for every enterprise in that sub-graph the cost of that enterprise is higher than the total investment by all players outside of that sub-graph investing in that enterprise. In this case, the iterative process described in the proposition cannot reach the empty graph since no vertex in that sub-graph can be removed in the iteration. 
\end{proof}

\section{Additional Proofs}\label{sec:appendix-additional-proofs}

\subsection{Proofs for Section \ref{sec:characterization}}

\begin{proof} \emph{(Lemma \ref{thm:zero-collaterals-1}):}
Assume {\small $x_{ki} + \sum_{\substack{j \in A_k}}x_j \geq Z_k\big(1+\frac{1}{\alpha_k}\big)$}. 
Since $c_{ki}=0$, if player $i$ invests, then $U_{ki} = R_{ki}$.
Due to monotonicity, as stated in Proposition \ref{thm:monotonicity}, it is sufficient to show that if only the players  
$A_k \cup \{i\}$ invest, $R_{ki} \geq x_{ki}$ (monotonicity shows that  $R_{ki}$ can only increase with cooperation by more players). 
I.e., we need to show that 
{\small
$R_{ki} = \max \Big(0, x_{ki} (1+\alpha_k)\big[1 - \frac{Z_k}{x_{ki} + \sum_{\substack{j \in A_k}}x_{kj}}\big]\Big) \geq x_{ki}$}. 
From our assumption it follows that
{\small
$(1+\alpha_k)\big[1 - \frac{Z_k}{x_{ki} + \sum_{\substack{j \in A_k}}x_{kj}}\big] \geq$ 
$(1+\alpha_k)\big[1 - \frac{Z_k}{Z_k\left(1+1/\alpha_k\right)}\big] = 1$}, 
and hence $R_{ki} \geq x_{ki}$. Thus, player $i$ prefers cooperating.
Next, assume that only the players $A_k \cup \{i\}$ invest and $U_{ki} \geq x_{ki}$. 
Again, as $c_{ki}=0$, if player $i$ invests $U_{ki} = R_{ki}$. The enterprise return to player $i$ thus satisfies:
{\small  
$R_{ki} = x_{ki} (1+\alpha_k)\big[1 - \frac{Z_k}{x_{ki} + \sum_{\substack{j \in A_k}}x_{kj}}\big] \geq x_{ki}$}.
As $Z_k,\alpha_k > 0$, this leads to Equation (\ref{eq:zero-collateral-condition}). 
\end{proof}
\begin{proof} \emph{(Lemma \ref{full-collaterals}):}
Assume  {\small $x_{ki} + \sum_{\substack{j \in A_k}}x_{kj} \leq Z_k$}. The return from enterprise $k$ to player $i$ if she \emph{invests} in $k$ is: 
{\small $\max \big(0, x_{ki} (1+\alpha_k)\big(1 - \frac{Z_k}{x_{ki} + \sum_{\substack{j \in A_k}}x_{kj}}\big)\big) = 0$}. 
Player $i$'s utility then equals $c_{ki}$, and so $i$ prefers to cooperate only with a full collateral.
Next, assume by contradiction {\small $x_{ki} + \sum_{\substack{j \in A_k}}x_{kj} > Z_k$}. 
The enterprise return to $i$ is:  
{\small $R_{ki} = \max \big(0, x_{ki} (1+\alpha_k)\big(1 - \frac{Z_k}{x_{ki} + \sum_{\substack{j \in A_k}}x_{kj}}\big)\big) > 0$}. 
player $i$ has a non-zero return from $k$.
We can select {\small $c_{ki} = x_{ki} - R_{ki}$} and then {\small$U_{ki} = x_{ki}$}, a contradiction. Hence, Equation (\ref{eq:full-collateral-condition}) holds.
\end{proof}
\begin{proof} \emph{(Lemma \ref{large-alpha}):} 
Formally, for $j \ s.t. \ x_{kj} = 0$ the collateral is zero by definition, and the statement holds. Next, fix any solution $\textbf{c}$ to the collateral minimization problem. Note that such a solution is also a minimal collateral matrix. 
Denote by $K$ the set edges corresponding to investment opportunities in enterprise $k$ (outgoing edges of $k$), and by $F_k$ the 
set of all edges in $K$ in which $k$ gives a full collateral. Let $\sigma$ be an order of the edges of iterated elimination of dominated strategies. 
Consider the player $i$ such that $(k,i)$ is the first investment opportunity in $k$ (edge $(k,i)$) in the order $\sigma$. 
Denote by $A_i$ the set of \emph{invest edges} when all the edges until $(k,i)$ in the order $\sigma$ are \emph{cooperate edges} and all edges after $(k,i)$ in $\sigma$ are \emph{defect edges} and denote $A_{ki} = A_i \cap K$.
Player $i$ strictly prefers to cooperate w.r.t. $k$ 
and thus she cannot be in default 
(or otherwise $i$ would prefer to defect). 
If 
$i \in F_k$ 
then the statement holds for this edge ($i$ has a full collateral). 
Otherwise, if  
{\small $x_{ki} + 
\sum_{\substack{j \in A_k}}x_{kj} 
\leq Z_k$}, 
then, by Lemma \ref{full-collaterals}, 
$i$ must also get full collateral, a contradiction to $i \notin F_k$. 
Therefore, as the input is integer, {\small$x_{ki} + \sum_{\substack{j \in A_k}}x_{kj} \geq Z_k + 1$}, and so for $\alpha_k > Z_k$ it holds that {\small$x_{ki} + \sum_{\substack{j \in A_k}}x_{kj} \geq Z_k + 1 > Z_k\left(1+ 1 / \alpha_k\right)$}, which by Lemma \ref{thm:zero-collaterals-1} means that the minimum collateral $c_{ki}$ needed to incentivize player $i$ to cooperate is zero. 
The argument holds for all the following investments in $k$ in the order $\sigma$. 
Thus we conclude that in an optimal solution the collateral for every investment opportunity in enterprise $k$ is either full or zero.
\end{proof}

\subsection{Proofs for Section \ref{sec:single}}

\begin{proof} \emph{(Lemma \ref{thm:zero-collaterals-2}):} 
By Lemma \ref{thm:zero-collaterals-1}, $i$ having a strict best reply to cooperate implies  
{\small $x_i + \sum_{\substack{j \in A}}x_j \geq Z\left(1+\frac{1}{\alpha}\right)$}. 
After player $i$ joins and plays \emph{cooperate}, for every player 
$m \notin A \cup \left\{i\right\}$  it holds that 
{\small $x_m + x_i + \sum_{\substack{j \in A}}x_j \geq  Z\left(1+\frac{1}{\alpha}\right)$} 
and so, by Lemma \ref{thm:zero-collaterals-1} player $m$ has a strict best reply to cooperate with zero collateral. 
\end{proof}

\begin{proof} \emph{(Proposition \ref{thm:collateral-costs}):}
Fix any ordering $\sigma$ of the players. For $\sigma$ to be an order of iterated elimination of strictly dominated strategies leading to the \emph{all-cooperate} equilibrium we require that player $\sigma_1$ has a strictly dominant strategy to cooperate and that given that $\sigma_1$ cooperates, player $\sigma_2$'s strict best reply is \emph{cooperate} for any action profile of the players $\sigma_{j>2}$, and so forth. Using Equation (\ref{R_i(II)}) this leads to the following set of equations.  
{\small
\begin{equation*}\label{?}
\begin{aligned}
	c_{\sigma_i} + R_{\sigma_i} =
	c_{\sigma_i} + \max \Big(0, x_{\sigma_i} (1+\alpha)\big[1 - \frac{Z}{\sum_{\substack{{\sigma_j \leq {\sigma_i}}}}x_j}\big]\Big)  
	\geq 
	x_{\sigma_i} \\
	\Longleftrightarrow  \quad 
	c_{\sigma_i} \geq x_{\sigma_i} - \max \Big(0, x_{\sigma_i} (1+\alpha)\big[1 - \frac{Z}{\sum_{\substack{{\sigma_j \leq {\sigma_i}}}}x_j}\big]\Big) 
\end{aligned}
\end{equation*}
}
where the first equality is since the worst case enterprise return $R_{\sigma_i}$ to player $\sigma_i$ if she cooperates is when only the players that preceded her in the order of elimination and herself cooperate.   
Her reward in this case must be at least $x_i$ to make her prefer investing. 
By the assumption of non-negative collaterals we obtain:
{\small
$c_{\sigma_i} \geq x_{\sigma_i} \Big[1 - \max \Big(0, \min \Big(1, (1+\alpha)\big[1 - \frac{Z}{\sum_{\substack{{\sigma_j \leq {\sigma_i}}}}x_j}\big]\Big)\Big) \Big],$
} 
which is equivalent to: 
{\small
$c_{\sigma_i} \geq x_{\sigma_i} \cdot \max \Big(0, \min \Big(1, 1-(1+\alpha)\big[1-\frac{Z}{\sum_{\substack{\sigma_j \leq \sigma_i}}x_{\sigma_j}}\big]\Big)\Big),$
} 
with equality giving the minimal collateral for each player ${\sigma_i}$. 
\end{proof}

\begin{proof} \emph{(Theorem \ref{thm:opt-partial-collaterals}):} 
Equation (\ref{eq:opt-partial-collaterals}) is obtained directly from Proposition \ref{thm:collateral-costs} under the condition that in the order $\sigma$ of iterated elimination of dominated strategies, players in $A$ appear first (in an arbitrary order) and they are followed by the players in $B$, ordered in an increasing order of their indices.
In cases of subsets of players with equal investment amounts the order of iterated elimination $\sigma$ can always be permuted such that players in each such subset will be ordered in the increasing order of their indices. This operation does not change the sum of collaterals due to symmetry. 

Now consider an optimal solution in which $A$ is the set of players with full collaterals and $B = A^c$.    
If there is only one player in $B$ the condition of the ordering of players in $B$ is trivially satisfied. 
If there are at least two players in $B$, assume for the purpose of contradiction that the condition is not satisfied. Then, there are two players $i_l$, $i_s$ with investments $l > s$ that are consecutive players in the order of iterated elimination, i.e., $\sigma_{i_l} = \sigma_{i_s} + 1$.
We will show that for such pair of players one of the following holds: (i) the collaterals for both players are zero and they can be switched in the order $\sigma$ without affecting the sum of collaterals, or (ii) the sum of all collaterals can be reduced by switching the order of these two players in $\sigma$ and recalculating the collaterals by Proposition \ref{thm:collateral-costs}, which leads to a contradiction with the optimality assumption.

In the case that the collateral to player $i_s$ is zero, from Lemma \ref{thm:zero-collaterals-2} the collateral to player $i_l$ is also zero. Moreover, since $l > s$, it can be seen from Lemma \ref{thm:zero-collaterals-1} that it is possible to switch the players $i_s$ and $i_l$ in the order $\sigma$ and still the collateral to each of these two players in the minimal collateral vector calculated after the replacement remains zero. 
The remaining cases are the case where both players have positive collaterals and the case where the collateral to $i_s$ is positive and the collateral to $i_l$ is zero. 

Note that since the players $i_s$ and $i_l$ are consecutive players in the order $\sigma$, according to Proposition \ref{thm:collateral-costs}, the collaterals to other players that come before them or after them in $\sigma$ do not change if we switch between $i_s$ and $i_l$. Thus the only difference in the sum of collaterals from such a replacement is the difference in the total collaterals to these two players.  
Denote the sum of collaterals to the two players in the decreasing ordering $\sigma_{i_l} < \sigma_{i_s}$ by $cost(l,s)$ and the sum of collaterals to these two players in the order that we assumed by $cost(s,l)$, and denote the sum of investments of all the players $j, \ s.t. \ \sigma_j < \min(\sigma_{i_l}, \sigma_{i_s})$ as $X_0$. 

In case that both collaterals are positive the two costs are:
{\small
$cost(l,s) = l \big[1 - (1+\alpha)\big(1 - \frac{Z}{X_0+l}\big)\big] + s \big[1 - (1+\alpha)\big(1 - \frac{Z}{X_0+l+s}\big)\big]$,} and 
{\small
$cost(s,l) = s \big[1 - (1+\alpha)\big(1 - \frac{Z}{X_0+s}\big)\big] + l \big[1 - (1+\alpha)\big(1 - \frac{Z}{X_0+s+l}\big)\big]$.
} 
%
The difference is:
{\small
$cost(s,l) - cost(l,s) = \frac{Z(1+\alpha)(s+l+2X_0)sl}{(X_0+l)(X_0+s)(X_0+l+s)} > 0$.
} 
%

In the case that the collateral to the first player, $i_s$, is positive and the collateral to the second player, $i_l$, is zero, then:
{\small
$cost(s,l) = s \big[1 - (1+\alpha)\big(1 - \frac{Z}{X_0+s}\big)\big]$,
} and:  
{\small
$cost(l,s) = l \cdot \max\big(0,  \big[1 - (1+\alpha)\big(1 - \frac{Z}{X_0+l}\big)\big]\big)$.
} 
%
The difference, {\small$cost(s,l) - cost(l,s)$,} is:
{\small
$s \big[1 - (1+\alpha)\big(1 - \frac{Z}{X_0+s}\big)\big] > 0$,
} or: 
{\small  
$(l-s)\big[\alpha + \frac{Z(1+\alpha)}{(X_0+s)(X_0+l)}\big] > 0$.
} 
Thus, we can always reduce the sum of all collaterals, in a weak sense, by switching any pair of consecutive players in the order of iterated elimination who have partial collaterals such that the player with the smallest index (and so a larger or equal investment amount) comes first.  
\end{proof}

\begin{proof} \emph{(Corollary \ref{thm:decreasing-partial-collaterals-corollary}):}
As $x_a \geq x_b$, by the theorem we can assume in an optimal solution $a < b$. Hence, the per-dollar collaterals satisfy:
{\small
$\big[1 - (1+\alpha)\big(1-\frac{Z}{X_A + \sum_{\substack{j \leq a}}x_j}\big)\big] > \big[1 - (1+\alpha)\big(1-\frac{Z}{X_A + \sum_{\substack{j \leq b}}x_j}\big)\big]$.} I.e., the per-dollar collateral to player $a$ is higher. Since $a$'s investment amount is at least as high as $b$'s investment, both the absolute and per-dollar collateral to player $a$ are higher than these of player $b$.
\end{proof}

\begin{proof} \emph{(Lemma \ref{largest-player}):}
Fix any optimal solution to the collateral minimization problem. 
Denote by $S^*$ the set of players with full collaterals in that solution. The other players have zero collaterals 
by Lemma \ref{large-alpha}. 
ssume in contradiction that there is a player $i_l \in S^*$ with investment $l$ that is strictly larger than any investment not in $S^*$. Consider any player $i_s \notin S^*$. It has an investment $s<l$. 
Denote $X_0=(\sum_{i\in S^*} x_i) - l$. 
As $S^*$ is a solution, $i_s$ cooperates with zero collateral. 
By Lemma \ref{thm:zero-collaterals-1}: 
{\small $s + (X_0 + l) \geq  
Z (1+\frac{1}{\alpha})$.} 
I.e., 
{\small $l + (X_0 + s) \geq
Z (1+\frac{1}{\alpha})$.} 
Thus, if we switch the zero and full collaterals of $i_l$ and $i_s$, 
by Lemma \ref{thm:zero-collaterals-1} both players still prefer investing and by Lemma \ref{thm:zero-collaterals-2} the others follow, 
but the total collateral is $X_0 + s$,   
which is less than $X_0 + l$, 
in contradiction to the optimality assumption.
\end{proof}

\begin{proof} \emph{(Lemma \ref{knapsak-iff-inverse-knapsack-new}):}
	Recall that the (uniform) knapsack problem gets as input integers $\left\{x_i\right\}_{i=1}^n$ and an integer threshold $r$, and is asks to find a set $T$ of indecis such that $\sum_{i \in T} x_i$ is maximized, under the constraint that  
	$\sum_{i \in S}x_i \leq r $. Denote $X=\sum_{i\in [n]} x_i$. The set $T$ is a solution to the knapsack problem with threshold $r$ if and only if the set $K=[n]\setminus T$ is a solution to the inverse knapsack problem with threshold $t=X-(r+1)$. 
\end{proof}
\end{document}